\input epsf
\documentclass[11pt,letter]{article}
\usepackage[dvips]{graphicx}
\usepackage{a4wide}%,times}
\usepackage{epsfig} %,psfig}
\usepackage{subfigure}

\def\0{\mathbf{0}}

\def\beqa{\begin{eqnarray}}
\def\eeqa{\end{eqnarray}}
\def\beqas{\begin{eqnarray*}}
\def\eeqas{\end{eqnarray*}}

\newtheorem{theorem}{Theorem}[section]
\newtheorem{lemma}{Lemma}[section]

\newtheorem{corollary}{Corollary}[section]

\newtheorem{remark}{Remark}[section]

\newtheorem{observation}{Observation}[section]
\newtheorem{definition}{Definition}[section]
\newtheorem{assumption}{Assumption}[section]

\newcommand{\old}[1]{{}}

\def\endpf{{\ \hfill\hbox{\vrule width1.0ex height1.0ex}\parfillskip 0pt}}
\newenvironment{proof}{\noindent{\bf Proof:}}{\endpf}

\addtocounter{footnote}{0}

\title{Uniform unweighted set cover:\\ The power of non-oblivious local search %\footnote{A preliminary extended abstract of this paper appears in proceedings of WAOA 2007.}\\
}

\author{Asaf Levin\footnote{Chaya fellow. Faculty of Industrial Engineering and Management, Technion - Israel Institute of Technology, 32000 Haifa, Israel.
                  email: {\tt levinas@ie.technion.ac.il}}
         \and
         Uri Yovel\footnote{Faculty of Industrial Engineering and Management, Technion - Israel Institute of Technology, 32000 Haifa, Israel.
                            email: {\tt uyovel@tx.technion.ac.il} }}

\begin{document}
\maketitle

\begin{abstract}
We are given $n$ base elements and a finite collection of subsets of them. The size
of any subset varies between $p$ to $k$ ($p<k$). In addition, we assume that the
input contains all possible subsets of size $p$. Our objective is to find a subcollection
of minimum-cardinality which covers all the elements. This problem is known to be NP-hard.
We provide two approximation algorithms for it, one for the generic case, and an improved
one for the special case of  $(p,k) = (2,4)$.

The algorithm for the generic case is a greedy one, based on packing phases: at each phase
we pick a collection of disjoint subsets covering $i$ new elements, starting from $i=k$
down to $i=p+1$. At a final step we cover the remaining base elements by the subsets of
size $p$. We derive the exact performance guarantee  of this algorithm for all values of
$k$ and $p$, which is less than $H_k$, where $H_k$ is the $k$'th harmonic number.
However, the algorithm exhibits the known improvement methods over the greedy one
for the unweighted $k$-set cover problem (in which subset sizes are only restricted
not to exceed $k$), and hence it serves as a benchmark for our improved algorithm.

The improved algorithm for the special case of $(p,k) = (2,4)$ is based on non-oblivious local search:
it starts with a feasible cover, and then repeatedly tries to replace sets of size 3 and 4
so as to maximize an objective function which prefers big sets over small ones. For this
case, our generic algorithm achieves an asymptotic approximation ratio of $1.5 + \epsilon$, and
the local search algorithm achieves a better ratio, which is bounded by $1.458333...+ \epsilon$. 
\end{abstract}

\section{Introduction}
In the \emph{unweighted set cover problem}, we are given $n$ base elements and a
finite collection of subsets of them. Our objective is to find a {\it cover}, i.e.,
a subcollection of subsets which covers all the elements, of minimum-cardinality.
This problem has applications in diverse contexts such as efficient testing, statistical
design of experiments, crew scheduling for airlines, and it also arises as a subproblem
of many integer programming problems. For more information, see, e.g., \cite{Hoch}, Chapter 3.

When we consider instances of unweighted set cover such that each subset has at most $k$
elements, we obtain the \emph{unweighted $k$-set cover problem}.
This problem is known to be NP-complete
\cite{Ka}, and it is MAX SNP-hard for all $k \geq 3$ \cite{PY,CK,KMSV}.

\looseness=-1It is well known (see \cite{Ch}) that a greedy algorithm is an
$H_k$-approximation algorithm for unweighted $k$-set cover,
where $H_k=\sum_{i=1}^k {1 \over i}$ is the $k$'th harmonic
number and that this bound is tight \cite{Jo,Lo}.  For unbounded values
of $k$, Slav\'{\i}k \cite{Sl} showed that the approximation ratio of
the greedy algorithm for unweighted set cover is $\ln
n - \ln \ln n + \Theta (1)$. Feige \cite{Fe} proved that unless
$NP \subseteq DTIME(n^{polylog \ n})$, unweighted set cover
 cannot be approximated within a factor $(1-\epsilon )\ln
n$ for any $\epsilon >0$.  Raz and Safra \cite{RS} proved that if
$P\neq NP$, then for some constant $c$, unweighted set cover
cannot be approximated within a factor $c\log n$. This
result shows that the greedy algorithm is an asymptotically best
possible approximation algorithm for this problem (unless $NP \subseteq DTIME(n^{polylog \ n})$).
Goldschmidt, Hochbaum, and Yu
\cite{GHY} modified the greedy algorithm for unweighted
$k$-set cover and showed that the resulting algorithm has a
performance guarantee of $H_k-{1 \over 6}$.  Halld\'orsson
\cite{Halld1} presented an algorithm based on a local search that has an
approximation ratio of $H_k - {1 \over 3}$ for  unweighted
$k$-set cover and a ($1.4+\epsilon$)-approximation algorithm for
unweighted 3-set cover.  Duh and F\"urer \cite{DF} later
improved this result and presented an ($H_k-{1 \over
2}$)-approximation algorithm for unweighted $k$-set cover.
Levin \cite{Le06} improved their result and obtained  an
($H_k-0.5026$)-approximation algorithm for $k \ge 4$, and
Athanassopoulos et al. \cite{ACK07} presented a further improved
algorithm for $k \ge 6$ with approximation ratio approaching $H_k-0.5902$ for
large values of $k$.

All of these improvements
\cite{GHY,Halld1,DF,Le06,ACK07} are essentially the greedy algorithm, with modifications
on the way it handles small subsets. That is, they are all
based on running the greedy algorithm until
each new subset covers at most $t$ new elements (the specific value of $t$ depends
on the exact algorithm), and then use a different method to cover the remaining
base elements.

In \cite{HL}, Hochbaum and Levin consider the problem of covering the edges of a bipartite
graph $G$ using a minimum number of $K_{p,p}$ bicliques (which need not
be subgraphs of $G$). This problem arises in the context of optical networks design (see \cite{HL}),
where $p$ is typically $2$ or $3$.
In addition, it can be viewed as an instance of unweighted $p^2-$set
cover, where the base elements are $G$'s edges, and the input collection consists of
all $K_{p,p}$ graphs over $G$'s vertices. In that paper, they analyze the greedy
algorithm applied for this special case, and show that it returns a solution whose cost is
at most $(H_{p^2} - H_p + 1)OPT + 1$ (where $OPT$ is the optimal cost). They also present an
improved algorithm for the case $p=2$ based on the property of the bipartite graph $G$,
achieving an approximation ratio of $1.3 + \epsilon$.

If, in addition, the input collection contains some graphs that have
up to $k$ edges, $k > p$, then the resulting problem is an instance of the
 \emph{$(p,k)$-uniform unweighted set cover problem} (see \cite{HL}), which we
 denote by $(p,k)$-UUSC. That is,
it is the variant of unweighted set cover where the size of every subset
varies between $p$ to $k$ ($p<k$), and the input contains all possible subsets of size $p$.
In fact, their analysis of the greedy algorithm is for this generalization.
Thus, the algorithms for unweighted $k$-set cover
serve as a benchmark for our algorithms for this problem.
%In addition, from a practical point of view, this variant has many applications (see \cite{HL}).

Recall that the dual problem of  unweighted $k$-set cover is the
\emph{(maximum) unweighted $k$-set packing problem}:
We are given $n$ base elements and a collection of subsets of them. Our objective is to
find a {\it packing}, i.e., a subcollection of {\it disjoint} subsets, of maximum-cardinality.
The fractional version of unweighted set packing is the dual
linear program of the fractional version of unweighted set
cover.  The greedy algorithm for this problem, which returns any maximal
subcollection of subsets, achieves an approximation ratio of $1\over k$.
 Hurkens and Schrijver \cite{HS} proved that for unweighted $k$-set packing,
 a local search algorithm is a ${2 -\epsilon} \over k$-approximation algorithm.
Athanassopoulos et al. \cite{ACK07} use this local search algorithm in each of their
"packing phases", and then use the method of Duh and F\"urer \cite{DF} in a final phase.

The weighted $k$-set cover problem and the weighted $k$-set packing problem are defined analogously.  However, this time each set has a cost (in the set cover variant) or a profit (in the set packing variant) and the goal is to minimize the total cost or to maximize the total profit, respectively.  The greedy algorithm for the unweighted versions and the weighted versions have the same approximation guarantee (for each of the two problems).  Hassin and Levin \cite{HL06} improved the resulting approximation ratio for the weighted $k$-set cover problem for constant values of $k$, and Arkin and Hassin \cite{AH} improved the greedy algorithm for the weighted $k$-set packing problem.

The method of \emph{local search} has been widely used in many hard combinatorial optimization
problems. The idea is simple: start with an arbitrary (feasible) solution. At each step, search
a (relatively small) neighborhood for an improved solution. If such a solution is found,
replace the current solution with it. Repeat this procedure until the neighborhood (of the
current solution) contains no improving solutions. At this point, return the current solution,
which is \emph{locally optimal}, and terminate.
Observe that in order for this method to run in polynomial time,
each local change should be computable in polynomial time, and the number of iterations should be
polynomially bounded.

Local search algorithms are mainly used in the framework of metaheuristics, such
as simulated annealing, taboo search, genetic algorithms, etc. From a practical point of view, they are usually
very efficient and  achieve excellent results - the generated solutions are near optimal.
However, from a theoretical point of view, there is usually no guarantee on the their worst-case performance.
In the thorough survey \cite{Angel}, Angel reviews the main results on local search algorithms that have
a worst-case performance guarantee. See also Halld\'orsson \cite{Halld2} for applications of this
method to $k$-dimensional matching, $k$-set packing, and some variants on independent set, vertex cover, set cover and
graph coloring problems.

In \cite{KMSV}, Khanna et al. present the paradigm of \emph{non-oblivious} local search. The idea, as
they comment, has been implicitly used in some known algorithms such as interior-point methods.
In that paper, they define the formal general algorithm in the context of MAX SNP.
Then, they develop non-oblivious local search algorithms for MAX $k$-SAT, and for the problem MAX $k$-CSP
which they define, which is a generalization of all the problems in MAX SNP. The idea in the context of set cover
is as follows. Any standard (i.e., \emph{oblivious}) local search algorithm must explicitly have the same objective: minimizing the number of picked sets. (Different such algorithms may look at different neighborhoods).
However, a \emph{non-oblivious} local search algorithm may have a different objective function to direct
the search.

{\bf Paper overview.} In section~\ref{generalAlg}, we present an algorithm for
$(p,k)$-UUSC (for any values of $p,k$). This algorithm is based on
applying the best known approximation algorithm for set packing (described in \cite{HS}) in each of the
packing phases. For $(p,k)$-UUSC  where $p \ge 2$, this algorithm exhibits all previously
known methods to improve upon the greedy algorithm for unweighted $k-$set cover. Hence, this algorithm
serves as a benchmark for our improved algorithm. For the special case of $(p,k) = (2,4)$ it achieves an
asymptotic approximation ratio of $1.5 + \epsilon$.
In section \ref{alg24}, we present an improved algorithm for the case of $(p,k) = (2,4)$, which is based
on non-oblivious local search, and we show that its (absolute) approximation ratio is at most $\frac{35}{24} + \epsilon
=  1.458333...+ \epsilon$.
In section~\ref{concluding}, we discuss some open questions.

\section{A first approximation algorithm for $(p,k)$-UUSC} \label{generalAlg}
Our algorithm is described in Figure \ref{fig:A1alg}.

\begin{figure}[!h]
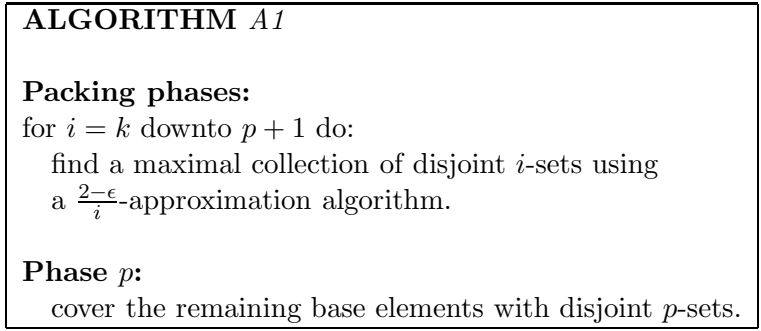
 \centering
\setlength{\parindent}{0pt}
\begin{tabular}{|l|}
\hline
\textbf{ALGORITHM {\it A1 } }\\
  \\
%\texttt{Input}: k-Sets. \\
%\texttt{Output}: a Set Cover.\\
%\textbf{begin}\\
\textbf{Packing phases:}\\
for $i=k$ downto $p+1$ do: \\
\quad find a maximal collection of disjoint $i$-sets using \\
\quad a ${2-\epsilon} \over i$-approximation algorithm. \\
   \\
\textbf{Phase $p$:}\\
\quad cover the remaining base elements with disjoint $p$-sets. \\
%\textbf{end.}\\
\hline
\end{tabular}
\caption{Algorithm {\it A1}.}
\label{fig:A1alg}
\end{figure}

We analyze this algorithm using a \emph{factor revealing linear program}.
We assume that $p \ge 1$, $k \ge 2$, $k>p$.
We also assume that the input satisfies the subset closure property and, consequently,
that the cover consists of disjoint subsets. Note that in explicit representation,
this causes the input size to increase by a factor of $2^k-1$ at the most,
since for each subset, all its non-trivial subsets are added to the collection.
However, such explicit representation is not necessary for our algorithm, and we
use it only for the analysis.
Another simplifying assumption for the analysis is:
\begin{assumption}  \label{apx_opt}
The input consists exclusively of the sets in $APX$ and $OPT$. In addition,
$APX \cap OPT = \emptyset$.
\end{assumption}
The justification of this assumption is fairly simple.
Regarding its first part, observe that
if the sets selected by $A1$ in phase $i$ cannot be improved, then this collection
of $i-$sets cannot be improved by replacing some of them by subsets of $OPT$ (or
subsets of them).
Hence, subsets outside $APX \cup OPT$ can be removed.

For the second part, observe that if there is a subset $S$ in both $APX$ and $OPT$, removing
$S$ and its elements from the input results in an instance for which $APX \setminus \{
S \}$ is a feasible solution and $OPT \setminus \{ S \}$ is an optimal
solution. But the approximation ratio for this new instance is
$\rho' \equiv \frac{|APX|-1}{|OPT|-1} \ge  \frac{|APX|}{|OPT|} \equiv
\rho$.

At any point in the execution of the algorithm, we define an $i-$set
to be a subset of size $i$, such that all of its elements are uncovered.
We define $a_{i,j}$ to be the ratio of the number of $j-$sets
in {\it OPT} in the beginning of packing phase $i$, to $|OPT|$, $i=p+1,...,k$, $j=1,...,i$,
and for phase $p$ we define $a_{p,p}$ to be the ratio of the number of
uncovered elements in the beginning of phase $p$, to $p|OPT|$.
% Observe that

Our analysis of Algorithm $A1$ is similar to that of \cite{ACK07}. In each
packing phase $i$ $(p+1 \le i \le k)$ we find a collection of
$i-$sets which is maximal. Therefore, in all of the next phases $j$
$(p \le j < i)$ there are no $i-$sets available. Similarly, in phase
$p$ there are no  $i-$sets available, $i>p$. Thus:
\begin{eqnarray}
\label{sumLE1a}
\sum_{j=1}^i a_{i,j}  & \le & 1 \ , \qquad     i=p+1,...,k \ ,  \\
 \label{sumLE1b}
a_{p,p} & \le & 1 \ . \\  \nonumber
\end{eqnarray}

Denote by $V_i$ the remaining uncovered elements in the beginning of phase $i$, $i=p,...,k$.
By definition of $a_{i,j}$, their number is $|V_i| = \sum_{j=1}^i j a_{i,j} |OPT|$.
In packing phase $i$, we pick $i-$sets that cover the elements in $V_i \setminus V_{i-1}$.
Since $V_{i-1} \subseteq V_i$ their number is:
\begin{equation}  \label{numOfDiffV}
|V_i \setminus V_{i-1}| = |V_i|- |V_{i-1}| = \left(   \sum_{j=1}^i j
a_{i,j} -   \sum_{j=1}^{i-1} j a_{i-1,j}       \right) |OPT| \quad , \quad
i=p+1,...,k \ .
\end{equation}
At the beginning of packing phase $i$, there are at least
$a_{i,i}|OPT|$ available $i-$sets. Therefore, the
$\frac{2-\epsilon}{i}-$ approximation algorithm picks at least
$(\frac{2-\epsilon}{i}) a_{i,i}|OPT|$ $i-$sets, thus covering at
least $(2-\epsilon) a_{i,i}|OPT|$ new elements. Hence, $|V_i
\setminus V_{i-1}| \ge (2-\epsilon) a_{i,i}|OPT|$. Using
(\ref{numOfDiffV}) and omitting the $\epsilon$ term, this yields:
\begin{equation}  \label{approxCons}
\sum_{j=1}^{i-1} j a_{i-1,j} - \sum_{j=1}^{i-1} j a_{i,j} - (i-2)a_{i,i}
 \le   0   \quad , \quad  i=p+1,...,k \ .
\end{equation}
Define $t_i$ to be the number of $i-$sets that are picked in packing phase $i$,
$i=p+1,...,k$. Then (\ref{numOfDiffV}) yields:
\begin{equation}  \label{t_i}
t_i = \frac{1}{i}  |V_i \setminus V_{i-1}|
= \left(   \frac{1}{i} \sum_{j=1}^i j a_{i,j} -  \frac{1}{i} \sum_{j=1}^{i-1} j a_{i-1,j}
 \right) |OPT|   \quad , \quad  i=p+1,...,k \ ,
\end{equation}
and for phase $p$ define $t_p$ as:
\begin{equation}  \label{t_p}
t_p = a_{p,p} |OPT| \ .
\end{equation}
Note that  $\lceil a_{p,p} |OPT| \rceil$ is the number of $p-$sets that are picked
and possibly an additional set of size less than $p$, covering the remaining elements.
Due to this last set, we obtain an \emph{asymptotic} approximation ratio. Specifically,
it is $\frac { \sum_{j=p}^{k} t_i + 1} {|OPT|}$. Using (\ref{t_i}),(\ref{t_p}), we obtain:
\begin{eqnarray}  \label{primalUB}
\frac { \sum_{j=p}^{k} t_i + 1} {|OPT|}  & = & a_{p,p} + \sum_{i=p+1}^k
\left(  \frac{1}{i} \sum_{j=1}^i j a_{i,j} -  \frac{1}{i}
\sum_{j=1}^{i-1} j a_{i-1,j}        \right) + \frac{1}{|OPT|} \nonumber \\
 & = & \frac{1}{k}
\sum_{j=1}^k j a_{k,j} + \sum_{i=p+1}^{k-1} \left( \frac{1}{i(i+1)}
\sum_{j=1}^i j a_{i,j} \right)  + \frac{1}{p+1} a_{p,p}  + \frac{1}{|OPT|} \ .
\end{eqnarray}

Thus, maximizing the right-hand side of (\ref{primalUB}) subject to
the constraints (\ref{sumLE1a}),(\ref{sumLE1b}),(\ref{approxCons}) and $a_{i,j} \ge
0$, yields an upper-bound on the approximation ratio of Algorithm
$A1$. Observe that asymptotically, the term $\frac{1}{|OPT|}$ is arbitrarily small. For convenience,
since it is a constant in the objective function, we omit it.
The resulting LP is: \\

\noindent {\bf Program} {\boldmath $(P)$}
\begin{eqnarray}
\max &  \frac{1}{k} \sum_{j=1}^k j a_{k,j}
+ \sum_{i=p+1}^{k-1} \left(  \frac{1}{i(i+1)} \sum_{j=1}^i j a_{i,j} \right)
 + \frac{1}{p+1} a_{p,p} & \nonumber\\
\label{primalCons1}
\textrm{s.t.} &  \sum_{j=1}^i a_{i,j} \le 1  \textrm{  }   \qquad \qquad  \qquad \qquad \qquad  \qquad   &  i=p+1,...,k  \\
\label{primalCons2}
& a_{p,p}  \le 1  \qquad \qquad \qquad \qquad \qquad \qquad \qquad  \quad & \\
\label{primalCons3}
&  \sum_{j=1}^{i-1} j a_{i-1,j} - \sum_{j=1}^{i-1} j a_{i,j} - (i-2)a_{i,i}
 \le   0  &  i=p+1,...,k  \\
&  a_{i,j} \ge 0  \qquad \qquad \qquad \qquad \qquad \qquad \qquad  \quad  &  i=p,...,k,  \quad  j=1,...,i \ . \nonumber
\end{eqnarray}
It is possible to derive a closed-form solution for this LP.
\begin{theorem} \label{LP_solution}
The solution of program $(P)$ is given by:
\begin{itemize}

\item{\bf Case 1:}
$k-p$ even:
\label{primalSolEvenDiff}
$a_{p+2j+1,p+2j}=a_{p+2j,p+2j}=1$  for all $j=0,...,\frac{k-p-2}{2}$,
$a_{k,k}=1$,
and all other $a_{i,j}$'s are zeros.

\item{\bf Case 2:}
$k-p$ odd:
\label{primalSolOddDiff}
$a_{p+2j+1,p+2j}=a_{p+2j,p+2j}=1$ for all $j=0,...,\frac{k-p-3}{2}$,
$a_{k,k} = a_{k-1,k-2} = 1$,
and all other $a_{i,j}$'s are zeros.
\end{itemize}
$A1$ is an asymptotic $(\rho + \epsilon)$-approximation algorithm for
$(p,k)$-UUSC, where $\rho$ is $(P)$'s objective function value, and is given by:
\begin{displaymath}
\rho = \left\{
\begin{array}{lll}
H_\frac{k}{2} - H_{\frac{p}{2}} + 1 &  p  \textrm{ even, }  k \textrm{ even}  \\
H_\frac{k-1}{2} - H_\frac{p}{2}   + 1 + \frac{1}{k} - \frac{1}{k(k-1)}
&  p  \textrm{ even, }  k \textrm{ odd} \\
2(H_{k} - H_{p+1}) - H_\frac{k}{2}  + H_\frac{p+1}{2} + 1 +  \frac{1}{k} - \frac{1}{k(k-1)}
&  p  \textrm{ odd, }  k \textrm{ even} \\
2(H_{k+1} - H_{p+1}) - H_\frac{k+1}{2}  + H_\frac{p+1}{2} + 1
&  p  \textrm{ odd, }  k \textrm{ odd} \ . \\
\end{array}  \right.
\end{displaymath}
\end{theorem}

\noindent
The proof is technical, and can be found in the Appendix. This is an asymptotic approximation ratio
due to the $\frac {1}{|OPT|}$ term which we neglected.
\begin{corollary}
$A1$ is an asymptotic $(1.5 + \epsilon)$-approximation algorithm for $(2,4)$-UUSC.
\end{corollary}

\section{An improved algorithm for $(2,4)$-UUSC }  \label{alg24}
In this section, we describe an improved algorithm for the case
$(p,k) = (2,4)$. That is, subsets' sizes are between $2$ to $4$, and
all possible $2-$sets are available.
Our algorithm is based on a \emph{non-oblivious} local search.
Specifically, denote by $X_2, X_3, X_4$, the
number of $2,3,4$-sets in $APX$, respectively. Then the number of
base elements is $n = 2X_2 + 3X_3 + 4X_4$ and the set cover objective is
to minimize $X_2 + X_3 + X_4$. However, the objective of our
algorithm is to {\it maximize } $4X_4 + X_3$. This is equivalent to
minimize $X_2 + X_3$. Intuitively, the large
sets are given higher priority because a cover which consists of
many large sets is good (due to the disjointness assumption).
Observe that this objective function is related to that of packing
problems, which are the dual of covering problems.
Our local search algorithm is described in Figure
\ref{fig:A2alg}. It is parameterized by $\epsilon$, which we assume
to be small enough, say $\epsilon \le  \frac{1}{100}$, and in addition, without loss of generality we assume that
$\frac{1}{\epsilon}$ is an integer.
\begin{figure}[!h] \centering
\setlength{\parindent}{0pt}
\begin{tabular}{|l|}
\hline
\textbf{ALGORITHM {\it A2 } }\\
  \\
%\texttt{Input}: k-Sets. \\
%\texttt{Output}: a Set Cover.\\
%\textbf{begin}\\
%\textbf{Packing Phases (i=p,...,k)}\\
%for $i=k$ downto $3$ do: \\
1. Start with an arbitrary feasible cover. \\
%   \\
2. Perform a {\it local search improvement step:} \\
 \qquad remove up to $\frac{1}{\epsilon}$ $3-$ and $4-$sets,\\
 \qquad insert any number of $3-$ and $4-$sets, so as to maximize $4X_4+X_3$. \\
3. Goto step 2, until no local search improvement step exists.\\
4. Cover the remaining base elements with $2-$sets. \\
\hline
\end{tabular}
\caption{Algorithm {\it A2}.}
\label{fig:A2alg}
\end{figure}

\noindent $APX$, the cover returned by the algorithm, is a
 {\it local optimum}. The following observation is trivial:
\begin{observation}   \label{apxOptSameOrder}
Every feasible solution $SOL$ is of size ${n \over 4} \le |SOL| \le {n \over 2}$.
Consequently, $|APX| = \Theta (|OPT|)$, $|OPT| = \Theta (|APX|)$.
\end{observation}

Note that this observation implies that if  $|OPT| \le  \frac{1}{2\epsilon}$,
then $APX$ is also an optimal solution. We use the following definition for convenience:
\begin{definition}
The $i-$sets in $OPT$ are called \emph{$i-$columns}; the $i-$sets in
$APX$ are called \emph{$i-$rows}. We  simply use \emph{columns} and
\emph{rows} in places where their size is irrelevant or clear from the context.
\end{definition}

\subsection{Restricting the input type}
In order to analyze the performance of Algorithm $A2$, we  assume, as in the previous
section, that the input collection satisfies the subset closure property, and that feasible
solutions consist of disjoint subsets.
We also continue to assume Assumption \ref{apx_opt},
i.e., that the input is $APX \cup OPT$, where $APX \cap OPT = \emptyset$.
The next assumption, which is less trivial, restricts the type of instance in the bad examples for the algorithm:
\begin{assumption} \label{opt4sets}
The instance belongs to one of the following two types:
\begin{itemize}
\item{Type A:}
$OPT$ consists exclusively of $4-$columns, $APX$ consists of $2-,3-$ and $4-$rows,
\item{Type B:}
$OPT$ consists exclusively of $3-$ and $4-$columns, $APX$ consists exclusively of $2-$ and $4-$rows.
\end{itemize}
\end{assumption}
In order to justify this assumption, we prove the following result:
\begin{lemma} \label{Lemma_opt4sets}
Let $I$ be a given instance. Let $APX$ be a local optimum in $I$, let $SOL$ be an arbitrary (feasible)
solution in $I$ with $|SOL| \le |APX|$, and let $\rho \equiv \frac{|APX|}{|SOL|}$.
Then there exists an instance $I'$ having solutions denoted by $SOL'$ and $APX'$, satisfying:
(i) $APX'$ is a local optimum in $I'$ achieving the same approximation ratio, i.e., $\rho' \equiv \frac{|APX'|}{|SOL'|}  =\rho$,
(ii) $SOL'$ contains no $2-$columns,
(iii) $SOL'$ contains no $3-$columns or $APX'$ contains no $3-$rows.
\end{lemma}
\begin{proof}
Recall that $APX \cap SOL = \emptyset$ by assumption. We refer to $SOL$'s sets as columns.
Given $I$, we  construct the new instance $I'$ in two phases. In Phase $1$ we
eliminate the $2-$columns in $SOL$ (if any); in Phase $2$ we  try to eliminate the $3-$columns in it.
We begin by describing Phase $1$.
Denote by $n^{APX}_{\{2,3\}}$ the number of $2-$ and $3-$rows in $APX$, and by
$n^{SOL}_2$  the number of $2-$columns in $SOL$.
We may assume that $n^{APX}_{\{2,3\}} \ge 1$, otherwise both $APX$ and $SOL$ are
optimal solution (consisting entirely of $4-$sets).
We  show how to eliminate $min\{  n^{APX}_{\{2,3\}}  ,  n^{SOL}_2  \}$ $2-$columns
from $SOL$.  Thus, if $n^{SOL}_2 > n^{APX}_{\{2,3\}}$ we may recursively apply this
transformation to the resulting new instance, until (the new) $SOL$ contains no $2-$columns.
In addition, the approximation ratio, $\rho$, remains the same.

Let $\mathcal{C}$ be a collection of $min\{  n^{APX}_{\{2,3\}}  ,  n^{SOL}_2  \}$ $2-$columns
in $SOL$ (if $n^{SOL}_2 \le n^{APX}_{\{2,3\}}$ then it is unique).
Then for each $2-$column $c \in \mathcal{C}$, there exists a distinct $2-$ or $3-$row in $APX$ which
we denote by $r_c$.
Let $I'$ be the instance in which each $c \in \mathcal{C}$ is extended to a $3-$column
$c' \equiv c \cup \{x_c\}$ and $r_c \in APX$ is extended to $r_c' \equiv r_c \cup \{x_c\}$,
where $x_c$ is a distinct new base element corresponding to $c$.
These extended sets will be referred to as {\it new}.
Sets from which new sets were obtained will be called {\it source sets}.

Construct from $APX$ a feasible solution for $I'$ by replacing each source row
by the new row extending it. Denote the resulting collection by $APX'$.
Similarly, construct $SOL'$ from $SOL$ by replacing each source column in $SOL$ by the new
column extending it. That is, $SOL'$ contains new $3-$columns obtained from source $2-$columns in $SOL$ ;
$APX'$ contains new $3-$ and $4-$rows obtained from source $2-$ and $3-$rows in $APX$.

We show that $APX'$ is a local optimum in $I'$.
Suppose to the contrary that this is not so. Then there exist a row collection $T' \subseteq APX'$,
and a subset collection $S'$ consisting of columns and (possibly, by the subset-closure assumption) of
sub-rows of $T'$ satisfying: (i) $|T'| \le \frac{1}{\epsilon}$, and (ii) replacing $T'$ by $S'$ improves
the objective function value.
More specifically, for $j \in \{2,3,4\}$, denote by
$t_j'$ and $s_j'$ the number of $j-$sets in $T'$ and $S'$, respectively.
Then by assumption:
\begin{equation}  \label{contraLO}
4t_4' + t_3' < 4s_4' + s_3' \ .
\end{equation}
Let $T \subseteq APX$ consist of the source ($2-$ and $3-$)rows from which the new ($3-$ and $4-$)rows in $T'$ were obtained, and of all the remaining non-new rows in $T'$.
Similarly, let $S$ consist of the source ($2-$)columns from which the new ($3-$)columns
in $S'$ were obtained, and of all the remaining non-new columns in $S'$.
Let $m_3,m_4$ be the number of new $3-$,$4-$rows in $T'$, respectively (i.e., $S$ has $m_3+m_4$
source ($2-$)columns which were extended to new ($3-$)columns in $S'$). Thus,
\begin{equation}  \label{s_t_treat2columns}
t_4' = t_4 + m_4, \  t_3' = t_3 + m_3 - m_4,
 \  s_4' = s_4, \  s_3' = s_3 + m_3 + m_4 \ .
\end{equation}
Using (\ref{contraLO}) and (\ref{s_t_treat2columns}), we obtain:
\begin{displaymath}
\begin{array}{lll}
4t_4 + t_3 & = & 4(t_4' - m_4) + t_3' -m_3 + m_4
 =  4t_4' + t_3' - m_3 - 3 m_4   \\
& < &  4s_4' + s_3' - m_3 - 3 m_4
= 4s_4 + s_3 - 2m_4 \le 4s_4 + s_3 \ ,  \\
\end{array}
\end{displaymath}
%4t_4 + t_3 = 4(t_4' - m_4) + t_3' -m_3 + m_4
%= 4t_4' + t_3' - m_3 - 3 m_4
%< 4s_4' + s_3' - m_3 - 3 m_4
%= 4s_4 + s_3 - 2m_4 \le 4s_4 + s_3,
that is, $4t_4 + t_3 < 4s_4 + s_3$. But this implies that the algorithm can replace $T$
by $S$ in $I$ and improve the objective function. This is a contradiction to $APX$ being a local
optimum in $I$. Finally, since $|APX'| = |APX|$ and $|SOL'| = |SOL|$, it follows that
$\rho' = \rho$. Thus, at the end of Phase $1$, properties (i),(ii) stated in the Lemma hold.

We now proceed to describe Phase $2$. The idea is similar to that of Phase $1$, but with two differences:
first, the new rows which are used to cover the new base elements in the new
($4-$)columns are only $4-$rows (extending $3-$rows in $APX$). (This is so because extending a $2-$row
in $APX$ to a $3-$row may result in a non-local optimum); second, let $n^{APX}_3$ ($n^{SOL}_3$) denote the
number of $3-$rows (columns) in $APX$ ($SOL$). Then this time, as opposed to what we did in Phase $1$,
if $n^{APX}_3 < n^{SOL}_3$, we cannot repeatedly perform the transformation on the new instance, since
it is possible for a local optimum to contain no $3-$rows. Thus, $SOL'$ - the new solution constructed
from $SOL$, is only guaranteed to have $min \{ n^{APX}_3 , n^{SOL}_3 \}$ less $3-$columns than $SOL$.

With a slight abuse of notation, we let $I$ denote the instance resulted from Phase $1$,
with $APX$ and $SOL$ its corresponding solutions,  and let  $I'$  denote the new instance which
we construct in this phase, with $APX'$ and $SOL'$ its corresponding solutions.

Let $\mathcal{C}$ be a collection of $min \{ n^{APX}_3 , n^{SOL}_3 \}$ $3-$columns
in $SOL$. Thus, for each $c \in \mathcal{C}$, there exists a distinct $3-$row in $APX$, denoted $r_c$.
Define $I'$ to be the instance in which each $c \in \mathcal{C}$ is extended to
the new $4-$column $c' \equiv c \cup \{x_c\}$ and $r_c \in APX$ is extended to the new
$4-$row $r_c' \equiv r_c \cup \{x_c\}$, for a new distinct element $x_c$.

As was done in Phase 1, construct $APX'$ ($SOL'$) from $APX$ ($SOL$) by replacing
source sets by the new sets extending them.
That is, $SOL'$ contains new $4-$columns extending source $3-$columns in $SOL$ ;
$APX'$ contains new $4-$rows extending source $3-$rows in $APX$.

We show that $APX'$ is a local optimum in $I'$.
If this is not the case, there exists $T' \subseteq APX'$, with
$|T'| \le \frac{1}{\epsilon}$ that can be replaced by a collection $S'$ consisting of columns
and subsets of rows, improving the objective function value.
That is, using the notation $t_j'$ and $s_j'$ from before, the inequality
(\ref{contraLO}) holds.

Let $T \subseteq APX$ ($S$) consist of the source $3-$sets in $APX$ ($SOL$) from
which the new sets in $T'$ $(S')$ were obtained from, and all the other non-new sets
in $T'$ ($S'$). Let $m$ be the number of new columns in $S'$, which is equal to that of the
new rows in $T'$. Thus,
\begin{equation}  \label{s_t}
t_4' = t_4+m, \ t_3' = t_3 - m,
 \    s_4' = s_4 + m,  s_3' = s_3 - m \ .
\end{equation}
Using (\ref{contraLO}) and (\ref{s_t}), we obtain:
\begin{displaymath}
4t_4 + t_3 = 4(t_4' - m) + t_3' + m
= 4t_4' + t_3' - 3 m
< 4s_4' + s_3' - 3 m
= 4s_4 + s_3 \ ,
\end{displaymath}
that is, $4t_4 + t_3 < 4s_4 + s_3$ - contradicting the fact the $APX$ is a local optimum in $I$.
Finally, we have $|APX'| = |APX|$, $|SOL'| = |SOL|$, implying that $\rho' = \rho$.

At the end of Phase $2$, the constructed instance $I'$ with its corresponding solutions
$APX'$ and $SOL'$ satisfy properties (i),(ii),(iii).
\end{proof}
  \\

Note that $4X_4 + X_3$, the objective function of Algorithm $A2$, does not take into account the number
of $2-$rows (as the algorithm only uses them to cover the remaining elements that were failed
to be covered by $3-$ or $4-$rows). This observation
motivates the following terminology, which we make solely for convenience: We will refer to
the base elements which are covered by $2-$rows as \emph{uncovered}.

Once again, we use a factor revealing LP to bound the approximation ratio of the algorithm.
That is, our goal is to formulate an LP whose objective function value is an upper
bound on the worst case approximation ratio of  $A2$ (denoted by $\rho$).
We treat each of the two instance types separately.

\subsection{Bounding $\rho$ in Type A-instances}
In this subsection we assume that the instance is of Type A, that is, $OPT$ consists
exclusively of $4-$columns, while there is no restriction on $APX$.
We  use the following notation:
\begin{definition}
For given $OPT$ and $APX$,
let $O_{i,j}$ be the set of columns in which $i$ elements are covered by $4-$rows
and $j$ elements are covered by $3-$rows, $0 \le i+j \le 4$, and let
$X_{i,j} \equiv \frac{|O_{i,j}|}{|OPT|}$ be the proportion of
$O_{i,j}$-columns in $OPT$.
\end{definition}
Observe that all $X_{i,j}$'s are non-negative and that they sum up to $1$. We would like to express
the objective function of set cover in terms of these new variables. We do so using a simple pricing
method: as each row of $APX$ costs $1$ and as the rows are disjoint, an element covered by an
$i-$row costs $1 \over i$, $i=2,3,4$. Thus, an $O_{i,j}$-column costs
\begin{equation} \label{cij}
c_{i,j} \equiv \frac{1}{4}i+\frac{1}{3}j+\frac{1}{2}(4-i-j) \ ,  \qquad      0 \le i+j \le 4 \ .
\end{equation}
Therefore:
\begin{displaymath}
|APX| = X_2+X_3+X_4  = \sum_{0 \le i+j \le 4} c_{i,j} |O_{i,j}|
= \sum_{0 \le i+j \le 4} c_{i,j} X_{i,j}|OPT| \ .
\end{displaymath}
Dividing by $|OPT|$ gives the approximation ratio of the given instance, which is
$\sum_{i,j} c_{i,j} X_{i,j}$. Thus, $\rho = \max_{\ I} \sum_{i,j} c_{i,j} X_{i,j}$
(the maximum taken over all legal instances), so our LP's objective is:
\begin{equation}  \label{objFunc}
\max \sum_{0 \le i+j \le 4} c_{i,j} X_{i,j} \ .
\end{equation}
In order to bound this function, we derive additional linear constraints. Our goal is to bound
the $X_{i,j}$'s with the highest $c_{i,j}$ coefficients. In light of our pricing scheme, this is
interpreted as not buying too many expensive columns. Starting by considering the most expensive ones,
the  following constraints are easy to establish:
\begin{lemma}  \label{zeroXij}
For any Type A-instance, $O_{0,0}, O_{0,1}, O_{0,2}, O_{0,3}, O_{1,0} = \emptyset$.
Equivalently, \\
$X_{0,0}, X_{0,1}, X_{0,2}, X_{0,3}, X_{1,0} = 0$.
\end{lemma}
\begin{proof}
Consider $O_{0,i}$, $i=0,...,3$. If, by contradiction, $O_{0,i} \ne \emptyset$
for some $i$, then there exists a column $S$ with $i$ of its
elements covered by $3-$rows, and the other elements are uncovered.
Removing these $3-$rows from $APX$ and inserting $S$ would increase
$A2$'s objective function. Thus, $O_{0,i} = \emptyset$.
If $O_{1,0} \ne \emptyset$ then there exists a column
$S$ having one element covered by a $4-$row, which we denote by $R$,
and the other elements are uncovered. Removing $R$ from $APX$, inserting $S$
and the $3-$row subset of $R$: $R \setminus (R \cap S)$
(recall the subset closure assumption), would again, increase $A2$'s
objective function. In either case we obtained a contradiction to $APX$ being a
local optimum.
\end{proof}  \\

Among the remaining variables, the two $X_{i,j}$'s which have the
largest coefficients in the objective function of the LP are,
according to (\ref{cij}),  $X_{1,1}$, with $c_{1,1}=\frac{19}{12}$,
and $X_{2,0}$, with $c_{2,0}=\frac{3}{2}$. We would like to obtain
an upper-bound on them, using a linear inequality. For this purpose,
we use an {\it intersection graph}.\\

\noindent {\bf The intersection graph $G$}\\
With a little abuse of terminology we will refer to $APX$, $OPT$, and
to subsets of them, as both the sets of indices representing the
subsets of base elements, and the sets of vertices representing them
in the following graph.

For a given instance, let $G = (V,E)$ be a bipartite graph, in which
one partite is the set of all $3-$ and $4-$row members of $APX$, and the second
partite is $OPT$. For $u$ a $3-$ or $4-$row in $APX$ and $v \in OPT$ there are
 $l$ (parallel) edges connecting $u$ and $v$ if the intersection of (the subsets
represented by) $u$ and $v$ consists of $l$ base elements.
Thus, for $v \in O_{i,j}$, $deg_G(v) = i+j$.
$G$ is {\it the intersection graph
corresponding to $APX$ and $OPT$}, or, {\it the intersection graph
of the given instance}, where $APX$ is a local optimum and $OPT$ is an
optimal solution of that instance.

Let $G$ be an intersection graph of a given instance, and
let $F = (V(F),E(F))$ be any induced subgraph of $G$. Denote by $O_{i,j}^F$ the columns in $F$ which are in $O_{i,j}$,
and denote by $n_r^F$ and $n_c^F$ the number
of rows and columns in $F$, respectively (i.e., $n_r^F \equiv |V(F) \cap APX|$,
$n_c^F \equiv |V(F) \cap OPT|$). Also let $n_c \equiv n_c^G$ , $n_r \equiv n_r^G$.
Note that $n_c = |OPT|$, and that $n_r \le |APX|$
(due to the uncovered elements, i.e., those covered by $2-$rows).
Finally, $F$ is called {\it small} if $n_r^F \le \frac{1}{\epsilon} - 2$,
otherwise it is called {\it big}.
(The reason for defining small subgraphs as those of size at most $\frac{1}{\epsilon} - 2$
rather than $\frac{1}{\epsilon}$ will be clear in the sequel).

Throughout the rest of the paper, we use 'CC' as an abbreviation for 'connected component'.
We  analyze the performance of Algorithm $A2$ by considering $G$'s CC's.
Recall that when we stated the algorithm, we observed that it is optimal for instances in which
an optimal solution consists of $\frac {1} {2\epsilon}$ sets at the most. In terms of $G$, this is generalized
to small CC's:
\begin{lemma}  \label{smallCC}
Let $G$ be an intersection graph of a given instance, and
let $F$ be a small CC of $G$.
Then the base elements covered by $F$'s columns are covered optimally by
Algorithm $A2$, and $|O_{4,0}^F| = n_c^F$, implying that
$|O_{i,j}^F| = 0$ for all $(i,j) \ne (4,0)$.
\end{lemma}
\begin{proof}
The algorithm, which has no access to $G$, performs local
improvement steps on collections of $3-$ and $4-$rows of size at most
$\frac{1}{\epsilon}$. Thus, it can remove all the $n_r^F \le \frac{1}{\epsilon}$
rows of $F$ and replace them with $F$'s columns, which optimally cover the base
elements in this CC. The rest of the claim follows from the fact that the instance is of Type A.
\end{proof} \\

Our goal is to upper-bound $A2$'s approximation ratio.
Since the following analysis can be performed componentwise on each of $G$'s CC's,
Lemma \ref{smallCC} implies that small CC's in $G$ can only improve the algorithm's performance,
decreasing its approximation ratio.
Thus, we may assume, without loss of generality:
\begin{assumption}  \label{Gconnected}
The intersection graph $G$ is connected and big.
\end{assumption}

We now turn to deal with $X_{2,0}$  and  $X_{1,1}$.
We derive a linear inequality in the $X_{i,j}$ variables which will be an
additional constraint in the LP that we construct. It is derived
using a special graph, which we construct in two stages.\\

\noindent {\bf The $H$ subgraph}\\
We define the following subgraph of $G$, which we refer to as {\it the $H$
subgraph}: it is the subgraph of $G$ induced by the set of $O_{1,1}$ and $O_{2,0}$
columns and the set of $4-$rows which intersect at least one $O_{1,1}$ or $O_{2,0}$ column.
See an example in Figure \ref{exampleH}.

\begin{figure} [t!]\vspace{-0.4in} %
\hspace{1.15in}
\begin{center}
\subfigure
[An instance ] %(only the $4-$rows intersecting $O_{2,0}\cup O_{1,1}$ are shown)
{ \includegraphics[angle=270,width=0.45\textwidth]
{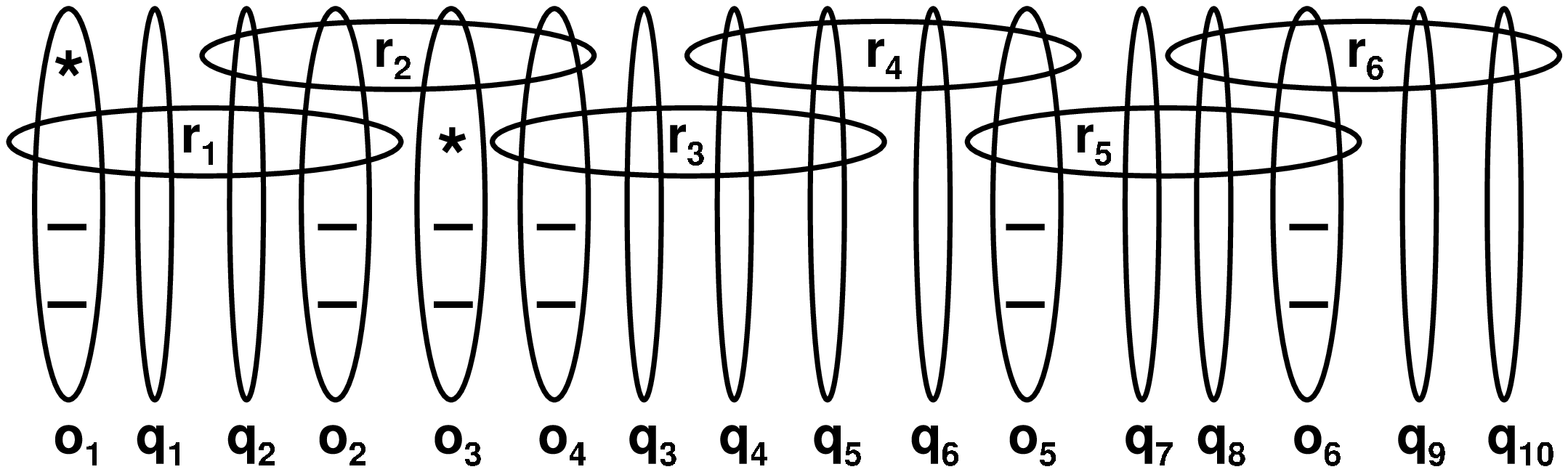}  }
\subfigure
[The corresponding $H$]{  \includegraphics[angle=270,width=0.45\textwidth]{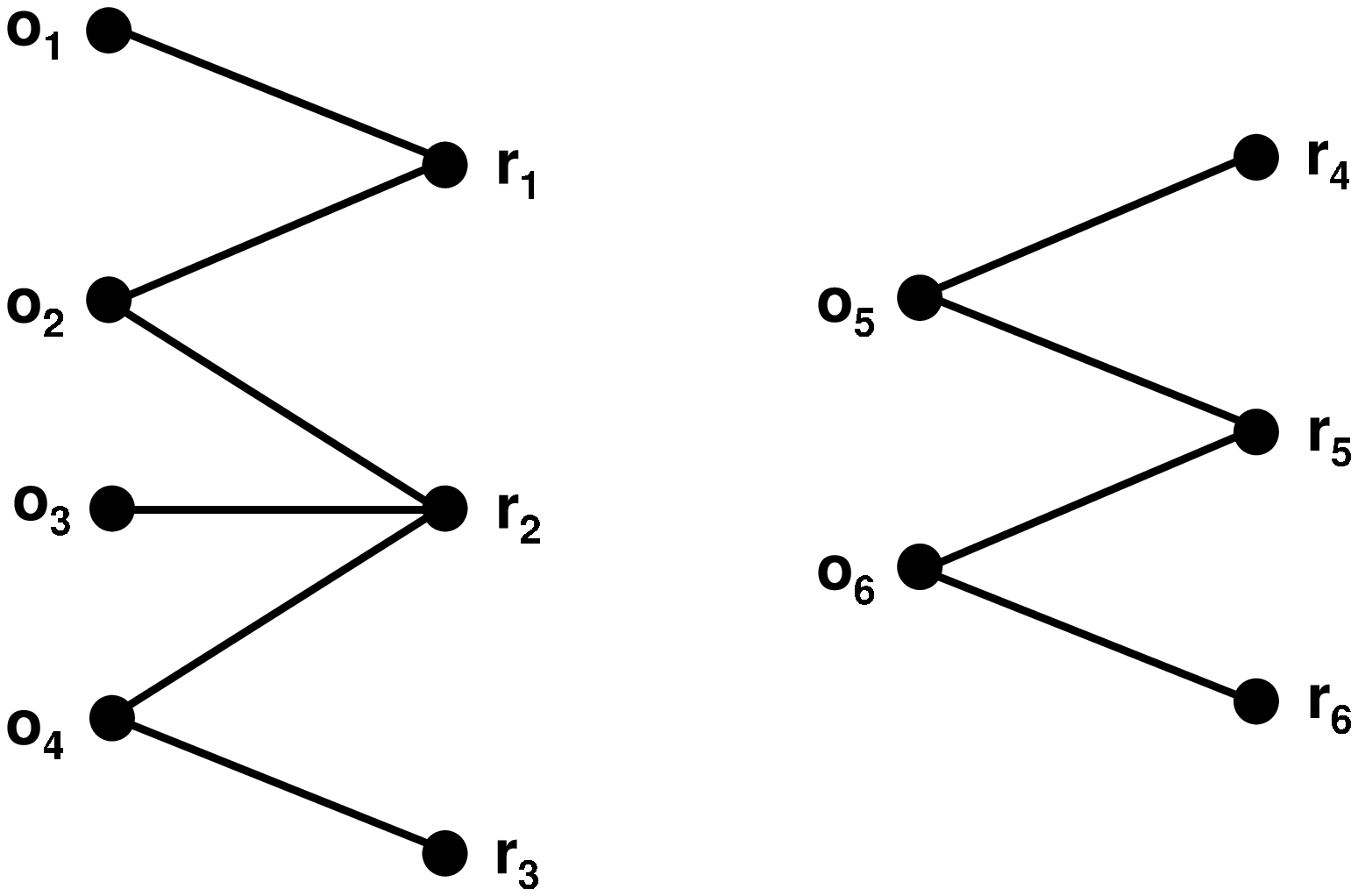}  }
%\hspace{0.5in}
\vspace{0.4in}
\caption{An example of an instance and the corresponding $H$ subgraph. In (a), the given
instance is shown, where only the $4-$rows intersecting $O_{2,0}\cup O_{1,1}$ are included.
A '-' stands for an uncovered base element, and a '*' stands for an element covered by a $3-$row.
Thus, $o_1 , o_3 \in O_{1,1}$, $o_2 , o_4 , o_5 , o_6 \in O_{2,0}$, and $q_1,...,q_9 \notin O_{1,1} \cup O_{2,0}$.
Also observe that $r_1 , r_5 \in R_{\{2\}}^H$, $r_3 , r_4 , r_6 \in R_{\{1\}}^H$, and $r_2 \in R_{\{3\}}^H$.}
\label{exampleH}
\end{center}
\end{figure}

Observe that $H$ need not be connected (as opposed to $G$, by Assumption \ref{Gconnected}).
Also observe that since the only rows in $H$ are $4-$rows, each $O_{1,1}$ vertex
has a single neighbor in $H$ (i.e., the $4-$row intersecting it). We record this fact for
future reference:
\begin{lemma}   \label{O11_leafH}
Each $O_{1,1}$ vertex is a leaf in $H$.
\end{lemma}
\noindent
For any subgraph $F$ of $H$, let $\Delta(F)$ denote the maximum degree of a
vertex in $F$. In addition, for $A \subseteq \{0,...,4\}$,
let $R_A^F$ denote the set of row vertices in $F$ of degree $i$ in $F$ for some
$i \in A$. % Note that from $H$'s definition, $n_c^F = |O_{1,1}^F| + |O_{2,0}^F|$.

We start by investigating the number of $O_{1,1}$ vertices in $H$. The following
result implies that there cannot be too many of them:
\begin{lemma}   \label{bigCC}
Let $c,d$ be two distinct column vertices in $O_{1,1}$ which belong to the same CC
of $H$. Then every $c-d$ path $P$ in $H$ has $n_r^P \ge \frac{1}{\epsilon}-1$
row vertices.
\end{lemma}
\begin{proof}
We may assume that $c$ and $d$ are connected by a (simple) path $P$ of minimum length among the paths
in $H$ connecting a pair of $O_{1,1}$ vertices.
Let $P \equiv (c,r_1,c_1,...,r_{l-1},c_{l-1},r_l,d)$.
By Lemma \ref{O11_leafH}, the $O_{1,1}$ vertices are leaves in $H$.
Therefore, the vertices $c_1,...,c_{l-1}$ are $O_{2,0}$ columns, and $r_1,...,r_l$ are
$4-$rows.
Denote by $r$ and $s$ the $3-$row neighbors in $G$ of $c$ and $d$,
respectively, and define $P' \equiv (r,c,r_1,c_1,...,r_{l-1},c_{l-1},r_l,d,s)$.
Observe that $P'$ cannot be a cycle: if $r=s$, then removing (the $3-$row) $r$
from $APX$, and inserting the two $3-$column subsets $c \setminus r_1$ (i.e.,
$c$'s two uncovered base elements and the singleton $c \cap r$)
and $d \setminus r_{l}$ increases the objective function by $1$,
which is a contradiction to $APX$ being a local optimum.
Thus, $P'$  is path from $r$ to $s$.
If $n_r^{P'} \le \frac{1}{\epsilon}$, the algorithm
can replace the rows $r,r_1,...,r_l,s$ with the
columns $c,c_1,...,c_{l-1},d$, again increasing its objective function,
which is a contradiction.
Thus, $n_r^{P'} \ge \frac{1}{\epsilon} +1$, implying
$n_r^P =  n_r^{P'} - 2 \ge \frac{1}{\epsilon} - 1$.
\end{proof}
\begin{corollary}  \label{smallCCOne11}
Every small CC of $H$ has at most one $O_{1,1}$ vertex.
\end{corollary}
As for big CCs, we have:
\begin{lemma}   \label{Voronoi}
Let $F$ be a big CC of $H$. Then $|O_{1,1}^F| = O(\epsilon)n_r^F$.
\end{lemma}
\begin{proof}
Assume that $|O_{1,1}^F| > 1$, otherwise the claim is trivial.
Construct a Voronoi diagram on the set of $F$'s vertices, with centers being
its $O_{1,1}$ columns. By Lemma \ref{bigCC}, any path connecting two distinct
such centers has at least $\frac{1}{\epsilon} - 1$ row vertices. Therefore, each Voronoi
cell contains at least $\lfloor \frac{1}{2}(\frac{1}{\epsilon} - 1) \rfloor$
vertices. Thus, $n_r^F \ge |O_{1,1}^F| \lfloor \frac{1}{2}(\frac{1}{\epsilon}
- 1) \rfloor$, implying $|O_{1,1}^F| = O(\epsilon)n_r^F$.
\end{proof}

Thus, the $O_{1,1}$ vertices are "negligible" in $H$, both in small and big components.
We proceed to investigate the number of $O_{2,0}$ vertices.
We specify two useful properties of $H$: the first states that
small CCs are either double edges (i.e., two parallel
edges between a pair of vertices), cycles, or trees,
and the second is a characterization of a local optimum.
\begin{lemma} \label{smallCC_H}
Every small CC of $H$ is either a double edge, a cycle, or a tree.
\end{lemma}
\begin{proof}
We  prove the claim by showing that a small CC of $H$ cannot include
a double edge or a cycle as a proper subset. Thus, any small CC which is
not a double edge or a cycle must be a tree.

We start by showing that two vertices that are connected by a double edge have
no other neighbors in $H$, implying that a CC of $H$ cannot include
a double edge as a proper subset.
Suppose that a column $c$ and a row $r$ are connected by a double edge.
Since, by Lemma \ref{O11_leafH}, the $O_{1,1}$ vertices are leaves, it follows that
$c \in O_{2,0}$. Thus, $|r \cap c|=2$ (i.e., $r$ covers two base elements of $c$),
and $c$ has no neighbors other than $r$.
So suppose to the contrary that $r$ has an additional neighbor $d \ne c$.
If $r$ covers two elements of $d$, then replacing $r$ with
$c,d$ produces a better solution, which is a contraction.
Otherwise, $|r \cap d| = 1$ and $d$ has an additional neighbor, which we denote by $s$.
Then removing $r$ and inserting $c$ and $d \setminus (d \cap s)$ (i.e., the $3$-row subset
of $d$ consisting of $d$'s two uncovered elements and the singleton $r \cap d$)
again produces a better solution, which is a contradiction.

%\begin{itemize}
%\item{Case 1: $d \in O_{1,1}$}\\ Denote by $s$ the $3-$row neighbor of $d$ in $G$ ($s$ is
%a not vertex in $H$). Then replacing $r,s$ with $c,d$ produces a better solution,
%which is a contraction.
%\item{Case 2: $d \in O_{2,0}$}\\ If $r$ covers two elements of $d$, then replacing $r$ with
%$c,d$ produces a better solution, which is a contraction.
%If $r$ covers one element of $d$, then replacing $r$ with $c$ and the $3-$susbset of $d$
%consisting of $d$'s two uncovered base elements and the singleton $d \cap r$
%again produces a better solution, which is a contradiction.
%\end{itemize}
In order to complete the proof, we  show that a small cycle has no neighbors
outside it, again, implying that a CC of $H$ cannot include it
as a proper subset.
Let $C$ be a small cycle in $H$. We show that for each vertex in $C$, its neighbors
in $H$ are precisely its two neighbors in $C$.
Again, since $O_{1,1}$ vertices are leaves (by Lemma \ref{O11_leafH}),
it follows that $C$'s vertices alternate between rows and $O_{2,0}$-columns. By
definition, each $O_{2,0}$ column has exactly two $(4-)$row
neighbors, hence, they are in $C$. As for the rows of $C$, suppose to the contrary that
there exists a $(4-)$row vertex $r \in C$ that has a neighbor $c \in H \setminus C$.
First observe that $r$ cannot be connected to $c$ by a double edge since in that case,
as we just proved, that double edge is by itself a CC, which is a contradiction.
Thus $r$ covers a single base element of $c$. Let $c'$ be the $3-$column subset of $c$
consisting of $c$'s two uncovered base elements and (the singleton) $r \cap c$.
As $n_r^C \le \frac{1}{\epsilon}$, the following local step
can be applied: remove $C$'s rows from the current solution and
insert $C$'s columns and $c'$.
The number of $4-$sets in the new
solution is the same, while the number of $3-$sets increases by one.
Thus, this step is a local improvement one, which is a contradiction.
\end{proof}\\
\begin{lemma} \label{rowLeavesInT}
Let $T$ be a small subtree of $H$.
(i) If all the leaves in $T$ are $(4-)$row vertices, then their number, $|R_{\{1\}}^T|$, is at most $4$.
(ii) If $T$ has exactly $|R_{\{1\}}^T| = 4$ leaves, then $T$ contains no $O_{1,1}$ vertices.
\end{lemma}
\begin{proof}
Assume $\Delta(T) > 2$, otherwise the claim is trivial (note that for part (ii),
if $\Delta(T) \le 2$ then $T$ cannot have $4$ leaves).
Hence $R_{\{3,4\}}^T \ne \emptyset$. \\
(i)
Since $T$'s vertices alternate between rows and
columns, it follows that
\begin{equation}  \label{nr_nc1}
n_r^T = n_c^T + 1 \ .
\end{equation}
To see this, partition
$T$ into edge-disjoint paths by the following iterative procedure:
start with any path $P$ connecting two arbitrary leaves, and mark
its vertices. Clearly, $n_r^P = n_c^P + 1$. As long as there exist
unmarked vertices, choose a minimal (with respect to inclusion) path $Q$ connecting an
unmarked leaf to
a marked vertex $u$. Note that $u \in R_{\{3,4\}}^T$, i.e., $n_r^Q =
n_c^Q+1$, and since $Q$ is minimal, all of $Q$'s vertices except for
$u$ are unmarked. Marking $Q$'s vertices, the number of row vertices
which are marked for the first time is equal to the number of
such column vertices. Summing over all paths, we obtain $n_r^T =
n_c^T + 1$.

Observe that for each row leaf $r \in R_{\{1\}}^T$, $r$'s neighbor is a column in $O_{2,0}$,
since the $O_{1,1}$ are leaves (by Lemma \ref{O11_leafH}) and $\Delta(T) > 2$ by assumption.
As $n_r^T \le \frac{1}{\epsilon}$, the following local step can be
applied:
\begin{itemize}
\item remove the $n_r^T$ rows of $T$ from the current solution,
\item insert the $n_c^T$ columns of $T$,
\item for each ($4-$row) leaf $r \in R_{\{1\}}^T$, insert its $3-$row subset consisting of the
three elements which are not covered by $r$'s ($O_{2,0}$) neighbor in
$T$.
\end{itemize}
Thus, we traded one $4-$row for $|R_{\{1\}}^T|$ $3-$sets. Due to our
objective function, we must have $|R_{\{1\}}^T| \le 4$, otherwise
this step would be a local improvement one, which is a
contradiction. \\

(ii)
Suppose to the contrary that there exists a subtree $T \subseteq H$ with $|R_{\{1\}}^T| = 4$ row
leaves  such that $O_{1,1}^F \ne \emptyset$.
Denote these row leaves by $r_1,...,r_4$, and let $c_1,...,c_4$ be their corresponding
neighbors. Observe that $c_1,...,c_4 \notin O_{1,1}$, since if $c_i \in O_{1,1}$ for
some $i$, then by Lemma \ref{O11_leafH} it is a leaf, implying that $(r_i,c_i)$ is an
isolated edge, contradicting the assumption that $T$ is a tree with four leaves.
By Corollary \ref{smallCCOne11}, there is exactly
one $O_{1,1}^F$ vertex, which we denote by $c$. Let $r$ be $c$'s $3-$row neighbor,
and let $T' \equiv (V(T) \cup \{ r \} , E(T) \cup \{ (c,r) \})$. Thus, all the
leaves in $T'$ are row vertices. It then follows, by exactly the same argument in part (i),
that $n_r^{T'} = n_c^{T'} + 1$. Therefore, in $T$ we have: $n_r^T = n_c^T$.
As $n_r^{T'} \le \frac{1}{\epsilon}$, we can remove the $n_r^T$ $4-$rows and (the $3-$row) $r$,
and insert the $n_c^T$ columns and the four $3-$row subsets of the leaves: $r_i \setminus c_i$, $i=1,...,4$.
The number of $4-$sets remain the same, while the number of $3-$sets increases by $3$, which
is a contradiction.
\end{proof}
  \\

\noindent
We emphasize that $T$ need not be a CC of $H$. It may be a proper subset
of a CC. If $T$ is a CC, the following result holds:
\begin{corollary}  \label{corTreeR1R34}
Let $T$ be a small CC of $H$ which is a tree. Then $|R_{\{1\}}^T| \le 4$, $|R_{\{3,4\}}^T| \le 2$.
Consequently, $|R_{\{2\}}^T| \ge n_r^T - 6$.
\end{corollary}
\begin{proof}
The leaves of $T$ are either $R_{\{1\}}^T$ or $O_{1,1}$ vertices. If all of them are $R_{\{1\}}^T$ vertices,
then by Lemma \ref{rowLeavesInT} (i): $|R_{\{1\}}^T| \le 4$. Otherwise, Corollary \ref{smallCCOne11} implies
that $T$ contains exactly one $O_{1,1}$ vertex. Deleting it from $T$, we obtain a subtree $T'$ whose all leaves
are the $R_{\{1\}}^{T'}$ vertices. By Lemma \ref{rowLeavesInT} (i): $|R_{\{1\}}^T| = |R_{\{1\}}^{T'}| \le 4$.

For the second part,
recall from Graph Theory that the number of leaves in
a nontrivial connected graph $G$ with $n_i$ vertices of degree $i$,
$i=1,...,\Delta(G)$, is bounded by:
\begin{equation}  \label{graphTheoryIneq}
n_1 \le 2 + \sum_{i=3}^{\Delta(G)} (i-2)n_i \ .
\end{equation}
(This follows from $\sum_{i=1}^{\Delta(G)} i n_i = 2 |E(G)| \ge
2(|V(G)|-1)= 2(\sum_{i=1}^{\Delta(G)} n_i - 1)$).
If $G$ is a tree, then (\ref{graphTheoryIneq}) holds as an equality, which we apply to $T$ and obtain:
\begin{equation}  \label{graphTheoryEq}
|R_{\{1\}}^T| = 2 + |R_{\{3\}}^T| + 2 |R_{\{4\}}^T| \ .
\end{equation}
We then conclude that:
\begin{displaymath}
|R_{\{3,4\}}^T| = |R_{\{3\}}^T| + |R_{\{4\}}^T| \le |R_{\{3\}}^T| + 2 |R_{\{4\}}^T| = |R_{\{1\}}^T| -2 \le 2 \ .
\end{displaymath}
(The last inequality follows from the first part). Thus  $|R_{\{3,4\}}^T| \le 2$.
\end{proof}  \\

Corollary \ref{corTreeR1R34} implies that most of the rows in a small tree $T$ have degree $2$, i.e., they
are the $R_{\{2\}}^T$ vertices. This is intuitive, as we can view these rows as "links" connecting two
columns in a "chain", while very few rows  are "end-rows" (namely, the $R_{\{1\}}^T$ ones), and even fewer rows
are "links" to other "chains" (the $R_{\{3,4\}}^T$ ones). Observe that for $F$ a cycle or a double edge, it is
trivial that \emph{all} the rows have degree $2$, i.e.
%\begin{lemma} \label{allRowsInCycleDeg2}
%For a CC $F$ of $H$ which is either a cycle or a double edge:
$|R_{\{1\}}^F| = |R_{\{3,4\}}^F| = 0$, $|R_{\{2\}}^F| = n_r^F$.
%\end{lemma}

For the big CCs of $H$, the dominance of the $R_{\{2\}}$ rows still holds, but in a weaker sense.
In order to establish it, we  look at small neighborhoods around
the vertices of a big CC $F$, bound the number of vertices of degrees $3$ or $4$,
and by summation obtain a bound on $|R_{\{3,4\}}^F|$. A bound on
$|R_{\{1\}}^F|$ then follows naturally.\\

\begin{definition}
For $\epsilon > 0$ and $v \in H$, let $B_\epsilon(v)$ be the neighborhood of radius $\frac{1}{5\epsilon}$
centered at $v$ in $H$, i.e., the set of all vertices $u$ in $H$
such there exists a $u-v$ path in $H$ of length at most $\frac{1}{5\epsilon}$.
\end{definition}
Observe that $|B_\epsilon(v)|$ may be greater than
$\frac{1}{\epsilon}$. In addition, it is possible for a "boundary" vertex $u \in B_\epsilon(v)$
that $deg_{B_\epsilon(v)}(u) < deg_H(u)$,
i.e., if its distance from $v$ is exactly $\frac{1}{5\epsilon}$.\\
\noindent
\begin{lemma}   \label{R34inNeighborhoodHle2}
For any  $v \in H$, $B_\epsilon(v)$ contains at most two vertices of degree $3$ or $4$ in $H$, i.e.,
$|B_\epsilon(v) \cap R_{\{3,4\}}^H| \le 2$.
\end{lemma}
\begin{proof}
Suppose to the contrary that there exists $v \in H$ such that
$B_\epsilon(v)$ contains at least $3$ vertices in $R_{\{3,4\}}^H$. Pick any
three of these vertices and denote them by $v_1, v_2, v_3$. Let $B'$ be a spanning tree of $B_\epsilon(v)$, and let
$P_i$ be the $v-v_i$ path in $B'$, $i=1,2,3$. Let $T$ be the subtree of
$B'$ defined by  $T \equiv \bigcup_{i=1}^3 P_i$. We first show
how to augment $T$ to obtain a subtree $T' \subseteq  B_\epsilon(v)$ with
at least $5$ leaves which are column vertices:
\begin{itemize}
\item{Case 1:} $T$ has at least two leaves in $\{ v_1,v_2,v_3 \}$,
say $v_1$ and $v_2$. Then each of
$v_1, v_2$ has at least two neighbors which are not in $T$, and in
addition, either $v$ is a leaf or $v_3$ has at least one neighbor which is not in $T$.
These neighbors are distinct and are different from $v$, otherwise $H$ contains a small cycle as a proper
subset, contradicting Lemma \ref{smallCC_H}. Let $T'$ be the tree
obtained by adding the edges connecting these neighbors to $T$. Then
$T'$ has at least $5$ leaves, which are columns.
\item{Case 2:}  $T$ is a simple path from $v$ to (say) $v_1$: then $v_1$ has at least two neighbors
which are not in $T$, and each of $v_2,v_3$ has at least one
neighbor which is not in $T$. These neighbors are distinct by an
argument similar to that in Case 1. Let $T'$ be the tree obtained by
adding the edges connecting these neighbors to $T$. Then again, $T'$
has at least $5$ column leaves ($v$ being one of them).
\end{itemize}
In both cases, we obtained a tree $T'$ of size at most
$\frac{3}{5\epsilon}+5$, which we clearly may assume to be less than
$\frac{1}{\epsilon}$, with at least $5$ leaves which are column
vertices. Now, since $T'$ is small, it follows by Corollary \ref{smallCCOne11}
that among these $5$ column leaves, at most one is an $O_{1,1}$ column vertex.
Thus, at least $4$ leaves are $O_{2,0}$ columns. Each such $O_{2,0}$ leaf
has an additional row neighbor outside $T'$. Again, these neighbors are distinct,
otherwise there is a contradiction to Lemma \ref{smallCC_H}. Adding the edges
connecting these row neighbors to $T'$, we obtain a tree of size at
most $\frac{3}{5\epsilon}+10 \le \frac{1}{\epsilon}$.
It either has $5$ or more row leaves, or exactly $4$ row leaves and one
$O_{1,1}$ leaf. In both cases we obtain a contradiction to Lemma \ref{rowLeavesInT}.
\end{proof}   \\

\noindent
We are now ready to upper-bound the number of vertices of degree
$1,3$, and $4$ in the big CCs of $H$. In particular, this establishes
the dominance (in terms of a lower bound) of rows of degree $2$ which we previously stated.
Since, as we mentioned, we look at each CC separately,
all bounds are in terms of the total number of rows in the specific CC.
\begin{lemma}   \label{WAY2boundRbig}
Let $F$ be a big CC of $H$.
(i) $|R_{\{3,4\}}^F| = O(\epsilon)n_r^F$,
 (ii) $|R_{\{1\}}^F| = O(\epsilon) n_r^F$,
(iii) $|R_{\{2\}}^F| \ge (1 - O(\epsilon)) n_r^F$.
\end{lemma}
\begin{proof}
For part (i), observe that:
\begin{displaymath}
|R_{\{3,4\}}^F| = \sum_{v \in R_{\{3,4\}}^F} 1
\le \sum_{v \in R_{\{3,4\}}^F} 5 \epsilon |B_\epsilon(v)| \ ,
\end{displaymath}
where the inequality follows from the fact that for $v$ in a big CC, $|B_\epsilon(v)| \ge \frac{1}{5\epsilon}$.

Now, consider the {\it multi-set} of vertices which belong to the (possibly
overlapping) neighborhoods around all of $R_{\{3,4\}}^F$ vertices, that is, we look at
$S \equiv \cup_{v \in R_{\{3,4\}}^F} B_\epsilon(v)$ where we allow repetitions of
elements in $S$.
Every vertex appears at
most twice in $S$. To see this, suppose to the contrary that
there is a vertex $u$ which appears at least three times in $S$. Then any
three centers of neighborhoods which cover $u$ are three
$R_{\{3,4\}}^F$ vertices in $B_\epsilon(u)$. $F$ is a CC of $H$, therefore
$R_{\{3,4\}}^F \subseteq R_{\{3,4\}}^H$, implying
$|B_\epsilon(u) \cap R_{\{3,4\}}^H| \ge |B_\epsilon(u) \cap R_{\{3,4\}}^F| \ge 3$,
which is a contradiction to
Lemma \ref{R34inNeighborhoodHle2}. Hence, $\sum_{v \in R_{\{3,4\}}^F}
|B_\epsilon(v)| = |S| \le 2n_F$.
Combining this with the previous inequality, we obtain:
\begin{equation}   \label{boundR34big_n_F}
|R_{\{3,4\}}^F| \le  10 \epsilon n_F \ .
\end{equation}
We would like to obtain the bound in terms of $n_r^F$, the number of rows in $F$.
Observe that each column intersects at most $4$ rows. Thus, $n_c^F \le 4n_r^F$, implying
that  $n_F = n_r^F + n_c^F \le 5 n_r^F$.
Substituting this in (\ref{boundR34big_n_F}), we obtain:
\begin{equation}   \label{WAY2boundR34big}
|R_{\{3,4\}}^F| \le 50 \epsilon n_r^F \ .
\end{equation}
This proves part (i).

For part (ii), applying (\ref{graphTheoryIneq}) to $F$, we obtain:
\begin{displaymath}
|R_{\{1\}}^F| \le 2 + |R_{\{3\}}^F| + 2 |R_{\{4\}}^F| \le 2 + 2 |R_{\{3,4\}}^F|
\le 2 + 100 \epsilon n_r^F \le 102 \epsilon n_r^F \ ,
\end{displaymath}
where the third inequality follows from (\ref{WAY2boundR34big}), and the last one from the
assumption that $F$ is big. This proves part (ii).

Part (iii) follows from parts (i) and (ii) (as $n_r^F = |R_{\{1\}}^F| + |R_{\{2\}}^F| + |R_{\{3,4\}}^F|$).
%OLD PROOF:
%For part (iii), we have:
%\begin{displaymath}
%|R_{\{2\}}^F| = n_r^F - |R_{\{1\}}^F| - |R_{\{3,4\}}^F|
%\ge n_r^F - 152 \epsilon n_r^F,
%\end{displaymath}
%where the inequality follows from (\ref{WAY2boundR34big}) and (\ref{WAY2boundR1big}).
This completes the proof.
\end{proof}  \\

Now consider the $O_{2,0}^F$ columns for some big CC $F$ of $H$. We show that their number
is about the same as that of $R_{\{2\}}^F$ vertices. Intuitively, this is true since, as
we proved, most of $F$'s columns are in $O_{2,0}^F$, most of its rows are in $R_{\{2\}}^F$, and
in every path the vertices alternate between rows and columns.
%intuitive, as we can view these columns as "links"
%connecting two rows in a "chain", while very few rows (the
%$R_{\{3,4\}}^F$ ones) are linked to other "chains", and very few
%rows (the $R_{\{1\}}^F$ ones) are "end-rows".
Formally:
\begin{lemma} \label{boundO20}
For a big CC $F$ of $H$:
$| R_{\{2\}}^F | - O(\epsilon) n_r^F \le  |O_{2,0}^F| \le |R_{\{2\}}^F | + O(\epsilon) n_r^F$.
\end{lemma}
\begin{proof}
Let $F$ be a big CC of $H$.
We bound $\sum_{i=1}^4  i| R_{\{i\}}^F |$  from below and from above to obtain:
\begin{displaymath}
2|R_{\{2\}}^F| \le \sum_{i=1}^4  i| R_{\{i\}}^F |
\le | R_{\{1\}}^F | + 2| R_{\{2\}}^F| + 4| R_{\{3,4\}}^F|
\le 2| R_{\{2\}}^F | + O(\epsilon) n_r^F \ ,
\end{displaymath}
where the last inequality follows from Lemma \ref{WAY2boundRbig}(i),(ii).
Counting $F$'s edges using each of its two partite sets, we obtain:
\begin{displaymath}
\sum_{i=1}^4  i| R_{\{i\}}^F | = |O_{1,1}^F| + 2|O_{2,0}^F| = 2|O_{2,0}^F| + O(\epsilon) n_r^F \ ,
\end{displaymath}
where the last equality is by Lemma \ref{Voronoi}. Thus:
%Substituting this in the above inequality (ASAF: SHOULD IT BE CALLED inequalitIES ?), we obtain:
\begin{displaymath}
2|R_{\{2\}}^F| \le 2|O_{2,0}^F| + O(\epsilon) n_r^F
\le 2| R_{\{2\}}^F | + O(\epsilon) n_r^F \ .
\end{displaymath}
Subtracting $O(\epsilon) n_r^F$ from all sides and dividing by $2$ yields the claim. \\
\end{proof}  \\

We note that for small CC's, the last result holds in a stronger sense:
\begin{remark}  \label{O20negSmallCCH}
Let $F$ be a CC of $H$.
(i) If $F$ is a small tree then $| R_{\{2\}}^F | -1 \le |O_{2,0}^F| \le | R_{\{2\}}^F | + 5$.
(ii) If $F$ is a small cycle or a double edge then $|O_{2,0}^F| = |R_{\{2\}}^F |$.
\end{remark}
\begin{proof}
(i) Let $F$ be a small CC of $H$ which is a tree. First suppose that $F$ contains no $O_{1,1}$ vertices,
i.e., all of its leaves the $R_{\{1\}}^F$ vertices. Then equality (\ref{nr_nc1}) holds for $F$, i.e.,
$n_r^F = n_c^F + 1$ (this is true by the argument used in the proof of Lemma \ref{rowLeavesInT} (i)).
Thus, $|O_{2,0}^F| = n_c^F = n_r^F - 1$.
We now bound $|O_{2,0}^F|$ from above and from below:
\begin{displaymath}
|O_{2,0}^F|  = n_r^F - 1 \le | R_{\{2\}}^F | +5 \ ,
\end{displaymath}
where the inequality follows from Corollary \ref{corTreeR1R34}, and trivially:
\begin{displaymath}
|O_{2,0}^F|  = n_r^F - 1 \ge | R_{\{2\}}^F | -1 \ .
\end{displaymath}
Thus, $| R_{\{2\}}^F | -1 \le |O_{2,0}^F| \le | R_{\{2\}}^F | + 5$, as required.

If $O_{1,1}^F \ne \emptyset$ then by Corollary \ref{smallCCOne11}, $F$ contains exactly one
$O_{1,1}$ vertex. We delete it from $F$ to obtain a tree, denoted $F'$, with all its leaves being $R_{\{1\}}$
vertices. Thus, the last result holds for $F'$, i.e., $| R_{\{2\}}^{F'} | -1 \le |O_{2,0}^{F'}| \le | R_{\{2\}}^{F'} | + 5$. By observing that $O_{2,0}^F = O_{2,0}^{F'}$ and $R_{\{2\}}^F = R_{\{2\}}^{F'}$, we establish the
result for $F$ as well. \\

(ii) This is trivial.
\end{proof}  \\

Recall that our goal is to bound $X_{1,1}$ and $X_{2,0}$ - the proportions of $O_{1,1}$ and $O_{2,0}$
in $G$. So far we obtained a good estimation of their proportions in $H$: Corollary \ref{smallCCOne11} and Lemma
\ref{Voronoi} imply that the $O_{1,1}$ vertices are negligible in small and big CCs of $H$, respectively ;
Lemma \ref{boundO20} and Remark \ref{O20negSmallCCH}  imply that intuitively, the proportion of $O_{2,0}$
in $H$ is about one half (the other half consists mainly of rows of degree $2$).
However, in order to bound the proportions \emph{in $G$}, we need to take into account the columns which are
not in $O_{1,1}$ or $O_{2,0}$ but intersect some row in that CC. This motivates the following construction:  \\

\noindent {\bf The $\widetilde H$ graph}\\
Denote those columns which intersect some row in $H$ and
are not in $O_{1,1} \cup O_{2,0}$  by $\widetilde O$.
We construct the $\widetilde H$ graph, which need not be a subgraph of $G$, in two steps.
First, let $\widetilde H$ be the graph obtained from $H$ by connecting each row in $H$ to distinct
new vertices representing the $\widetilde O$-columns intersecting it. We emphasize that an $\widetilde O$-column
may appear in a few CCs of $\widetilde H$: Suppose $o \in \widetilde O$ and $C_1,...,C_l$, $2 \le l \le 4$ are CCs of $H$ such that each
one contains a row intersecting $o$. Then each corresponding CC in $\widetilde H$ will have its
own (distinct) vertex representing $o$. Thus, in terms of vertex labels (where each vertex has a label of the set
represented by it), the CCs of $\widetilde H$ are not $\widetilde O$-column-disjoint. But they
are $O_{1,1} \cup O_{2,0}$-column-disjoint as well as row-disjoint, and hence in particular
are edge-disjoint. Figure \ref{exampleTildeH} shows the $\widetilde H$ graph corresponding to the instance
given in Figure \ref{exampleH} up to this step.
\begin{figure} [t!]\vspace{-0.4in} %
\hspace{1.15in}
\includegraphics[angle=270,width=0.6\textwidth]{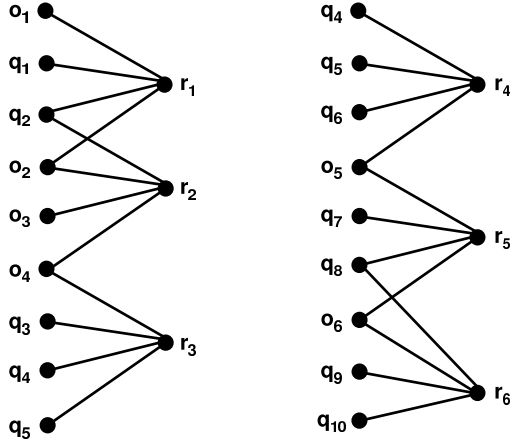}\vspace{0.4in}
\begin{center} \caption{The $\widetilde H$ graph corresponding to the instance
given in Figure \ref{exampleH} prior to the addition of the $\widetilde H_0$ subgraph.} \label{exampleTildeH}
\end{center}
\end{figure}

At a final step in the construction, we add to $\widetilde H$ the subgraph of $G$ induced
by the set of remaining vertices (if any), that is, all vertices which do not belong
to any CC from the previous step. Denote this subgraph by $\widetilde H_0$. Note that
$\widetilde H_0$ need not be connected. Let $\mathcal{C}$ denote the collection of $\widetilde H$'s
CCs. The disjointness of the rows implies:
\begin{equation}   \label{sumRowH}
n_r = n^{\widetilde H}_r = \sum_{F \in \mathcal{C} \cup \widetilde H_0 } n^F_r \ .
\end{equation}
We use the following notation.
Let $F \in \mathcal{C} \cup \widetilde H_0$.
Denote by $\widetilde O^F$ the $\widetilde O$-columns in $F$.
Denote by $E_{2,0}^F$ the set of edges incident to the $O_{2,0}^F$ columns,
by $E_{1,1}^F$ the set of edges incident to the $O_{1,1}^F$ columns,
and by $\widetilde E^F$ the set of remaining edges, i.e. $\widetilde E^F
\equiv E(F) \setminus ( E_{2,0}^F \cup E_{1,1}^F )$.
Note that for $F \in \mathcal{C}$ ($F \ne \widetilde H_0$),
all the edges in $\widetilde E^F$ are incident to columns in $\widetilde O^F$.
Define $E_{2,0} \equiv \bigcup_{F \in \mathcal{C}} \widetilde E_{2,0}^F$, $E_{1,1} \equiv \bigcup_{F \in \mathcal{C}} \widetilde E_{1,1}^F$, and $\widetilde E \equiv \bigcup_{F \in \mathcal{C}} \widetilde E^F$.
Finally, denote by $r_i^F$ the number of rows in $V(F) \cap R_{\{i\}}^H$.
The following observations are trivial (the first two hold in $H$ as well):
\begin{lemma}  \label{trivialH0properties}
(i) For each $F \in \mathcal{C} \cup \widetilde H_0$,  $|E_{1,1}^F| = |O_{1,1}^F|$ and $|E_{2,0}^F| = 2 |O_{2,0}^F|$,
(ii) For each $F \in \mathcal{C}$, each vertex $o \in O_{1,1}^F$ is a leaf in $F$,
(iii) For each $F \in \mathcal{C}$, a row in $F$ which belongs to $R_{\{i\}}^H$ contributes $i$
edges to $E_{1,1}^F \cup E_{2,0}^F$ and $4-i$ edges to $\widetilde E^F$, $i=1,...,4$, that is:
\begin{eqnarray} \label{cycleE20}
|E_{2,0}^F| + |E_{1,1}^F| & = & r_1^F + 2 r_2^F + 3 r_3^F + 4 r_4^F \ ,  \\
 \label{cycleEtilde}
 |\widetilde E^F| & = & 3 r_1^F + 2 r_2^F + r_3^F \ .
\end{eqnarray}
\end{lemma}
\begin{lemma} \label{lem_e_tilde_le_long_sum}
For any Type A-instance,
\begin{equation} \label{ineq_e_tilde_le_long_sum}
|\widetilde E| \le   |O_{1,2}| + |O_{1,3}|
+ 2|O_{2,1}| + 2|O_{2,2}| + 3|O_{3,0}| + 3|O_{3,1}| + 4|O_{4,0}| \ .
\end{equation}
\end{lemma}
\begin{proof}
Consider a vertex $o \in O_{i,j}$, $i \ge 1$, $(i,j) \notin \{(1,1),(2,0)\}$.
$o \in \widetilde O$ if there exist an $O_{2,0}$-column $q$ and a $4-$row $r$
such that $r \cap q \ne \emptyset$ and $r \cap o \ne \emptyset$.
In this case, $o$ contributes at most $i$ edges (possibly in different CCs) to
$\widetilde E$. (Otherwise it contributes zero).
\end{proof} \\

We now derive a linear inequality, which provides an upper-bound on the number of edges
in $E_{2,0}$ and $E_{1,1}$.
\begin{lemma}  \label{E20Etilde}
For any Type A-instance, $|E_{2,0}| + 3|E_{1,1}|  \le |\widetilde E| + O(\epsilon)n_r$.
\end{lemma}
\begin{proof}
We  show that for every $F \in \mathcal{C} \cup \widetilde H_0$, we have:
\begin{equation}   \label{E20F}
|E_{2,0}^F| + 3|E_{1,1}^F|  \le |\widetilde E^F| + O(\epsilon)n^F_r \ .
\end{equation}
By definitions of $E_{2,0}$, $E_{1,1}$ and $\widetilde E$, and using (\ref{sumRowH}), the claim then
follows by summing over all $F \in \mathcal{C} \cup \widetilde H_0$.
We distinguish four cases, according to the type of $F$.
\begin{itemize}
\item{\bf Case 1: $F = \widetilde H_0$ } \\
Since $\widetilde H_0$ contains no $O_{1,1}$ and $O_{2,0}$
vertices, it follows that $E_{1,1}^{\widetilde H_0} = E_{2,0}^{\widetilde H_0} = \emptyset$.
Hence (\ref{E20F}) trivially holds.
\item{\bf Case 2: $F$ is a CC of $\widetilde H$ obtained from a double edge or a
small cycle in $H$ } \\
Denote by $C$ the double edge or the cycle in $H$ from which $F$ is obtained.
Then $C$ is a CC of $H$, and its rows are precisely the rows of $F$ (because $\widetilde H$
was obtained from $H$ by adding columns).
$C$ has even length, with its vertices alternating between $O_{2,0}$ columns and $R_{\{2\}}^H$ $(4-)$rows.
Thus, each such row has two $O_{2,0}$ column neighbors (in $C$ and therefore in $F$) and two $\widetilde O$ column
neighbors (in $F \setminus C$). Therefore, it contributes two edges to $E_{2,0}^F$ and two to $\widetilde E^F$, i.e.:
\begin{equation} \label{E20cycle}
|E_{2,0}^F| = |\widetilde E^F| = 2n_r^F \ .
\end{equation}
Finally, observe that $O_{1,1}^F = \emptyset$: this is true because as we just noted, $C$'s columns
are only in $O_{2,0}$, and $F$ was obtained from $C$ by adding $\widetilde O$ columns (which, by definition,
are not in $O_{1,1}$).
Thus, $|O_{1,1}^F| = 0$, implying that $|E_{1,1}| = 0$.
This fact and (\ref{E20cycle}) establish (\ref{E20F}).
\item{\bf Case 3: $F$ is a CC of $\widetilde H$ obtained from a small tree in $H$ } \\
The rows of $F$ are precisely the rows of the tree
in $H$ which $F$ is obtained from. Thus,
%by Lemma \ref{trivialH0properties} (iii):
%\begin{eqnarray} \label{cycleE20}
%|E_{2,0}^F| + |E_{1,1}^F| & = & r_1^F + 2 r_2^F + 3 r_3^F + 4 r_4^F \ ,  \\
% \label{cycleEtilde}
% |\widetilde E^F| & = & 3 r_1^F + 2 r_2^F + r_3^F.
%\end{eqnarray}
%Subtracting the first equality from the second, we obtain:
subtracting (\ref{cycleE20}) from (\ref{cycleEtilde}), we obtain:
\begin{displaymath}
|\widetilde E^F| - |E_{2,0}^F| - |E_{1,1}^F|  = 2(r_1^F - r_3^F - 2 r_4^F) = 4 \ ,
\end{displaymath}
where the last equality follows from
$r_1^F = 2 + r_3^F + 2 r_4^F$, which holds due to (\ref{graphTheoryEq}). This implies:
\begin{equation}   \label{diffE20EtildeCycle}
|E_{2,0}^F|  \le |\widetilde E^F| - 4 \ .
\end{equation}
Now, by Corollary \ref{smallCCOne11}, $F$ can have at most one $O_{1,1}$ vertex.
By Lemma \ref{trivialH0properties} (ii),
such a vertex is a leaf in $F$, implying that $|E_{1,1}^F| = |O_{1,1}^F| \le 1$. Combining
this with (\ref{diffE20EtildeCycle}), we obtain:
\begin{displaymath}
|E_{2,0}^F| + 3|E_{1,1}^F| \le |E_{2,0}^F| + 3 \le |\widetilde E^F| - 1 \ ,
\end{displaymath}
establishing (\ref{E20F}).
\item{\bf Case 4: $F$ is a CC of $\widetilde H$ obtained from a big CC of $H$ } \\
We have:
\begin{displaymath}
|\widetilde E^F| \ge 2 r_2^F  \ge 2|O_{2,0}^F| - O(\epsilon)n_r^F = |E_{2,0}^F| - O(\epsilon)n_r^F \ ,
\end{displaymath}
where the first inequality follows from (\ref{cycleEtilde}),
the second inequality follows from Lemma \ref{boundO20}, and the equality follows from
Lemma \ref{trivialH0properties} (i).
In order to complete the proof, it suffices to show that $|E_{1,1}^F| = O(\epsilon)n_r^F$.
%We will show this by showing that $|O_{1,1}^F| = O(\epsilon)n_r^F$.
Denote by $F'$ the (big) CC of $H$ from which $F$ is obtained.
By Lemma \ref{Voronoi}, we have $|O_{1,1}^{F'}| = O(\epsilon)n_r^{F'}$.
Since $O_{1,1}^{F'} = O_{1,1}^F$ and similarly, the rows of $F'$ are
precisely the rows of $F$, we also have: $|O_{1,1}^F| = O(\epsilon)n_r^F$.
By Lemma \ref{trivialH0properties} (i): $|E_{1,1}^F| = |O_{1,1}^F|$.
Hence $|E_{1,1}^F| = O(\epsilon)n_r^F$, as required.
\end{itemize}
\end{proof}

\noindent
We are now ready to bound a linear combination of $X_{2,0}$ and $X_{1,1}$:
\begin{lemma} \label{LP_ineq_Type A}
For any Type A-instance:\\
\begin{equation} \label{X20inequality}
2 X_{2,0} + 3 X_{1,1} \le X_{1,2} + X_{1,3} + 2 X_{2,1} + 2 X_{2,2}
+ 3 X_{3,0} + 3 X_{3,1} + 4 X_{4,0} + O(\epsilon) \ .
\end{equation}
\end{lemma}
\begin{proof}
From Lemma \ref{E20Etilde} we have: $|E_{2,0}| + 3|E_{1,1}|  \le |\widetilde E| + O(\epsilon)n_r$.
From Observation \ref{apxOptSameOrder} we obtain $n_r = \Theta(n_c)$, so
we also have: $|E_{2,0}| + 3|E_{1,1}|  \le |\widetilde E| + O(\epsilon)n_c$.
We would like to write this inequality in terms of the column sets.
By summation, Lemma \ref{trivialH0properties}(i) implies that
 $|E_{1,1}| = |O_{1,1}|$ and $|E_{2,0}| = 2|O_{2,0}|$.
Thus, we obtain:
\begin{equation} \label{O20Ef}
2|O_{2,0}| + 3|O_{1,1}|  \le |\widetilde E| + O(\epsilon)n_c \ .
\end{equation}
Using Lemma \ref{lem_e_tilde_le_long_sum}, we obtain:
\begin{displaymath}
2|O_{2,0}| + 3|O_{1,1}|  \le  |O_{1,2}| + |O_{1,3}|
+ 2|O_{2,1}| + 2|O_{2,2}| + 3|O_{3,0}| + 3|O_{3,1}| + 4|O_{4,0}| + O(\epsilon)n_c \ .
\end{displaymath}
Dividing both sides by $n_c = |OPT|$, we obtain the required inequality.
\end{proof}
   \\

By providing the last constraint, Lemma \ref{LP_ineq_Type A} concludes our construction of
the LP, which upper-bounds $\rho$ - the approximation ratio of the algorithm (for Type A-instances). Recall
that the other constraints are that the variables are non-negative and that their sum is $1$.
In addition, the variables $X_{0,0}, X_{0,1}, X_{0,2}, X_{0,3}, X_{1,0}$ are zero, by Lemma \ref{zeroXij}.
The objective function was stated in (\ref{objFunc}). Thus, the complete program is:
\begin{eqnarray}
\max &   \sum_{0 \le i+j \le 4} c_{i,j} X_{i,j} \qquad \qquad \qquad \qquad \qquad \qquad \qquad \qquad \qquad \qquad & \nonumber\\
\label{A2primalCons1}
\textrm{s.t.} & 3X_{1,1}-X_{1,2}-X_{1,3}+2X_{2,0}- 2X_{2,1}- 2X_{2,2} -3X_{3,0}-3X_{3,1}-4X_{4,0} \le O(\epsilon)   &  \qquad  \\
\label{A2primalCons2}
&  \sum_{0 \le i+j \le 4} X_{i,j}  =  1  \qquad \qquad \qquad \qquad \qquad \qquad \qquad \qquad \qquad \qquad \qquad  &   \nonumber  \\
&  X_{i,j} \ge 0  \qquad \qquad  0 \le i+j \le 4  \qquad \qquad \qquad \qquad \qquad \qquad \qquad  \quad  & \nonumber \\
& X_{0,0}, X_{0,1}, X_{0,2}, X_{0,3}, X_{1,0} = 0   \qquad \qquad \qquad \qquad \qquad \qquad \qquad \qquad \quad  & \nonumber \\
\nonumber
\end{eqnarray}
(inequality (\ref{A2primalCons1}) is obtained from (\ref{X20inequality}) by rearranging terms).
Specifically, given $\epsilon >0$, the ratio $\rho$ is upper-bounded by the objective function
value of the LP. We now turn to solve this program. We simplify it, first by omitting
the zero variables $X_{0,0}, X_{0,1}, X_{0,2}, X_{0,3}, X_{1,0}$.
Denote the set of (remaining) relevant indices by $I \equiv \{ (1,1),(1,2),(1,3),(2,0),(2,1),(2,2),
(3,0),(3,1),(4,0) \}$.
Next, since our goal is to solve the LP for arbitrarily small values of $\epsilon$,
we replace $O(\epsilon)$ in the constraint (\ref{A2primalCons1}) by zero.
Using (\ref{cij}) to obtain the explicit values of $c_{i,j}$'s, the modified LP is:
\begin{eqnarray}
\max &   \frac{19}{12}X_{1,1} + \frac{17}{12}X_{1,2} + \frac{5}{4}X_{1,3} + \frac{3}{2}X_{2,0}
+ \frac{4}{3}X_{2,1} + \frac{7}{6}X_{2,2} + \frac{5}{4}X_{3,0}
+ \frac{13}{12}X_{3,1}  + X_{4,0}  & \nonumber\\
\label{modCons1}
\textrm{s.t.} & 3X_{1,1}-X_{1,2}-X_{1,3}+2X_{2,0}- 2X_{2,1}- 2X_{2,2} -3X_{3,0}-3X_{3,1}-4X_{4,0} \le 0   &  \qquad  \\
\label{modCons2}
&  \sum_{(i,j) \in I} X_{i,j}  =  1  \qquad \qquad \qquad \qquad \qquad \qquad \qquad \qquad \qquad \qquad \qquad  &   \\
&  X_{i,j} \ge 0  \qquad \qquad  (i,j) \in I  \qquad \qquad \qquad \qquad \qquad \qquad \qquad  \quad  & \nonumber \\
\nonumber
\end{eqnarray}
In order to solve this LP, we use the dual program. Let $y,z$ be the dual variables corresponding to
constraints (\ref{modCons1}),(\ref{modCons2}), respectively. The dual program is then:\\
\begin{eqnarray*}
\min & z  \qquad \qquad \qquad \qquad \qquad \qquad   &  \\
\textrm{s.t.} & 3y+z \ge \frac{19}{12} \qquad \qquad \qquad \quad &   \\
   & -y+z \ge \frac{17}{12} \qquad   \qquad \qquad  \quad &  \\
  & -y+z \ge \frac{5}{4} \qquad   \qquad \qquad  \quad &  \\
  & 2y+z \ge \frac{3}{2} \qquad   \qquad \qquad  \quad &  \\
  & -2y+z \ge \frac{4}{3} \qquad   \qquad \qquad  \quad &  \\
  & -2y+z \ge \frac{7}{6} \qquad   \qquad \qquad  \quad &  \\
  & -3y+z \ge \frac{5}{4} \qquad   \qquad \qquad  \quad &  \\
  & -3y+z \ge \frac{13}{12} \qquad   \qquad \qquad  \quad &  \\
  & -4y+z \ge 1 \qquad   \qquad \qquad  \quad &  \\
 & y \ge 0  \qquad   \qquad \qquad  \quad &  \\
\end{eqnarray*}
Let $X^*$ be the vector consisting of $X_{1,1}=\frac{1}{4}, X_{1,2}=\frac{3}{4}$, and $X_{i,j} = 0$
for all $(i,j) \in I \setminus \{ (1,1),(1,2) \}$.
It is clear that $X^*$ is a feasible primal solution. The corresponding objective function
value is $c_{1,1}X_{1,1}^* + c_{1,2}X_{1,2}^* = \frac{19}{12} \cdot \frac{1}{4} +
 \frac{17}{12} \cdot \frac{3}{4} = \frac{35}{24}$.
Let $(y^*,z^*) \equiv (\frac{1}{24} , \frac{35}{24})$. It is straightforward to verify that
it is a feasible dual solution. The corresponding objective function value is $z^* = \frac{35}{24}$, which is
equal to that of the primal. Thus, from the duality theorem, we conclude that  $X^*$ and
$(y^*,z^*)$ are optimal solutions to the primal and dual programs, respectively. By the construction of
the (primal) LP, we conclude the following result:
\begin{theorem}  \label{apxRatioTypeA}
For Type A-instances,
$A2$ is a $(\rho + \epsilon)$-approximation algorithm for
$(2,4)$-UUSC, where $\rho \le \frac{35}{24} = 1.458333...$
\end{theorem}

\subsection{Bounding $\rho$ in Type B-instances}

In this subsection we assume that the instance is of Type B, that is, $OPT$ consists of $3-$ and
$4-$columns, and $APX$ consists of $2-$ and $4-$rows. We  use the analogous notation to that of the
previous section.
\begin{definition}
For given $OPT$ and $APX$,
let $O_i^4$ be the set of $4-$columns in which $i$ elements are covered (by $4-$rows), $i=0,...,4$,
and let $X_i^4 \equiv \frac{|O_i^4|}{|OPT|}$ be the proportion of these columns in $OPT$.
Similarly, let $O_i^3$ be the set of $3-$columns in which $i$ elements are covered (by $4-$rows),
$i=0,...,3$, and let $X_i^3 \equiv \frac{|O_i^3|}{|OPT|}$.
For any graph $F$, let $O_i^{s,F} \equiv O_i^s \cap V(F)$, $s=3,4$, $i=0,...,s$.
\end{definition}
The objective function of set cover in terms of these new variables is:
\begin{equation}
|APX| = X_2+X_3+X_4 = \sum_{i=0}^4 c_i^4 |O_i^4| +  \sum_{i=0}^3 c_i^3 |O_i^3|
\end{equation}
where
\begin{displaymath} \label{ci4}
c_i^4 \equiv \frac{i}{4} + \frac{4-i}{2} = 2 - \frac{i}{4} \ ,  \qquad  i=0,...,4 \ ,
\end{displaymath}
and
\begin{displaymath} \label{ci3}
c_i^3 \equiv \frac{i}{4} + \frac{3-i}{2} = \frac{3}{2} - \frac{i}{4} \ ,  \qquad  i=0,...,3 \ .
\end{displaymath}
Explicitly, the column costs are:
\begin{equation} \label{ci34values}
(c_0^4,...,c_4^4) = (2,1.75,1.5,1.25,1) \quad , \quad (c_0^3,...,c_3^3) = (1.5,1.25,1,0.75) \ .
\end{equation}
Observe that $c_i^4 = c_{i,0}$ from the previous section ($i=0,...,4$).
The objective function of our LP, which bounds $\rho$ from above, is:
\begin{equation}
\max \sum_{i=0}^4 c_i^4 X_i^4 +  \sum_{i=0}^3 c_i^3 X_i^3 \ .
\end{equation}
Considering the highest $c_i^j$'s (i.e., the costs of the most expensive columns),
the following result is analogous to Lemma \ref{zeroXij}  and therefore its proof is omitted:
\begin{lemma}  \label{zeroX34}
For any Type B-instance, $O_0^3, O_0^4, O_1^4 = \emptyset$.
Equivalently, $X_0^3, X_0^4, X_1^4 = 0$.
\end{lemma}
The next highest coefficient is $c_2^4 = 1.5$, so we derive a bound on $X_2^4$.
The intersection graph $G$ is defined exactly the same, and we assume that it is connected and big
(i.e., Assumption \ref{Gconnected} holds for this instance type as well).
Formally, it consists of $3-$ and $4-$columns in the $OPT$ partite, and $4-$rows in the $APX$
one. As for $H$ and $\widetilde H$: \\

\noindent {\bf The $H$ subgraph} \\
$H$ is the subgraph of $G$ induced
by the $O_2^4$-columns and the ($4-$)rows intersecting them. Note that these columns are analogous
to the $O_{2,0}$ columns of Type A-instance, while there is no analog to $O_{1,1}$ columns. Thus,
$H$'s structure is the same, that is, $H$ obtained from a Type B-instance is a special case of $H$
obtained from a Type A-instance, with no $O_{1,1}$ columns.
Thus, the results from the previous section hold trivially. Specifically, regarding the $H$ subgraph,
Lemmas \ref{O11_leafH}, \ref{bigCC}
and \ref{Voronoi} are irrelevant, Lemma \ref{smallCC_H} holds, Lemma \ref{rowLeavesInT} (i) holds (part (ii)
is irrelevant), Lemma \ref{R34inNeighborhoodHle2} holds, and Lemma \ref{WAY2boundRbig}
holds. The analog of Lemma \ref{boundO20} is:
\begin{lemma} \label{boundO20_in_TypeB}
For any Type B-instance, for each big CC F: $| R_{\{2\}}^F | - O(\epsilon) n_r^F \le  |O_2^{4,F}|
\le |R_{\{2\}}^F | + O(\epsilon) n_r^F.$
\end{lemma}

\noindent {\bf The $\widetilde H$ graph} \\
$\widetilde H$ is, again, similar to  $\widetilde H$ from the previous section, but with no columns
analogous to $O_{1,1}$. Specifically, let $\widetilde O$ be the set of columns which intersect
some row in $H$ (i.e., a $4-$row intersecting some $O_2^4$ column). For each CC $F$ of $H$, connect
each row in $F$ to distinct vertices representing the $\widetilde O$-columns intersecting it.
Denote these vertices by $\widetilde O^F$.
Let $\widetilde E$ be the set new edges used to connect those vertices.
Also, let $\widetilde H_0$ denote the subgraph of $G$ induced by the remaining vertices (which
include all the $3-$rows), and add it to $\widetilde H$.
Finally, let $E_2^{4,F}$, $\widetilde E^F$, $E_2^4$ denote the set of edges incident
to $O_2^{4,F}$, $\widetilde O^F$, $O_2^4$ vertices, respectively.
The analog of Lemma \ref{trivialH0properties} is (only parts (i) and (iii) are relevant):
\begin{lemma}  \label{trivialH0properties_TypeB}
(i) For each $F \in \mathcal{C} \cup  \widetilde H_0$,
\begin{equation} \label{TypeB_e2o}
|E_2^{4,F}| = 2 |O_2^{4,F}| \ ,
\end{equation}
(ii) For each $F \in \mathcal{C}$, a row in $F$ which belongs to $R_{\{i\}}^H$ contributes $i$
edges to $E_2^{4,F}$ and $4-i$ edges to $\widetilde E^F$, $i=1,...,4$.
\end{lemma}
The analogs of Lemmas \ref{lem_e_tilde_le_long_sum} and  \ref{E20Etilde} are, respectively:
\begin{lemma}  \label{lem_e_tilde_le_long_sum_TypeB}
For any Type B-instance,
\begin{equation} \label{ineq_e_tilde_le_long_sum_TypeB}
|\widetilde E| \le   |O_1^3| + 2|O_2^3| +
3|O_3^4| + 3|O_3^3| + 4|O_4^4| \ ,
\end{equation}
\end{lemma}
\begin{lemma}  \label{E20Etilde_TypeB}
For any Type B-instance,
\begin{equation} \label{TypeB_e_ePLUSo}
|E_2^4|  \le |\widetilde E| + O(\epsilon)n_r \ .
\end{equation}
\end{lemma}
(The proof of Lemma \ref{E20Etilde_TypeB} is identical to that of Lemma \ref{E20Etilde} with substituting
$E_2^{4,F}$ for $E_{2,0}^F$ and $\emptyset$ for $E_{1,1}^F$). \\

\noindent
Using (\ref{TypeB_e2o}) and summing over all $F \in \mathcal{C} \cup  \widetilde H_0$, we obtain:
\begin{equation} \label{E24_2O24_TypeB}
|E_2^4| = \sum_{F \in \mathcal{C} \cup  \widetilde H_0} |E_2^{4,F}|
=  \sum_{F \in \mathcal{C} \cup  \widetilde H_0} 2 |O_2^{4,F}| = 2 |O_2^4| \ .
\end{equation}
Now, substituting (\ref{E24_2O24_TypeB}) in the left-hand side of (\ref{TypeB_e_ePLUSo}),
and (\ref{ineq_e_tilde_le_long_sum_TypeB}) in its right-hand side,
and using $n_r = \Theta(n_c)$ (from Observation \ref{apxOptSameOrder}), we obtain:
\begin{displaymath} \label{final ineq_o_TypeB}
2 |O_2^4| \le   |O_1^3| + 2|O_2^3| +
3|O_3^4| + 3|O_3^3| + 4|O_4^4| + O(\epsilon)n_c \ .
\end{displaymath}
Dividing by $n_c = |OPT|$, we obtain the analog of Lemma \ref{LP_ineq_Type A}:
\begin{lemma} \label{LP_ineq_Type B}
For any Type B-instance:
\begin{displaymath}
2 X_2^4 \le   X_1^3 + 2X_2^3 +
3X_3^4 + 3X_3^3 + 4X_4^4 + O(\epsilon)n_c \ .
\end{displaymath}
\end{lemma}
Using (\ref{ci34values}), the inequality from Lemma \ref{LP_ineq_Type B},
and  substituting $X_0^3, X_0^4, X_1^4 = 0$ (by Lemma \ref{zeroX34}),
we obtain the following LP, which upper-bounds $\rho$ for Type B-instances:
\begin{eqnarray}
\max &  1.5 X_2^4 + 1.25 X_3^4 + X_4^4
+ 1.25 X_1^3 + X_2^3 + 0.75 X_3^3  & \nonumber\\
\textrm{s.t.} &
2 X_2^4  - 3X_3^4  - 4X_4^4 - X_1^3 - 2X_2^3 - 3X_3^3  \le 0  &  \qquad \nonumber  \\
 &  X_2^4 + X_3^4 + X_4^4 + X_1^3 + X_2^3 + X_3^3  = 1 \qquad  \quad  & \nonumber \\
&  X_2^4,X_3^4,X_4^4, X_1^3,X_2^3,X_3^3  \ge 0   \qquad  \quad  & \nonumber \\
\nonumber
\end{eqnarray}
The dual program is:
\begin{eqnarray*}
\min & z  \qquad \qquad \qquad \qquad \qquad \qquad   &  \\
\textrm{s.t.}   & 2y+z \ge 1.5 \qquad   \qquad \qquad  \quad &  \\
  & -3y+z \ge 1.25 \qquad   \qquad \qquad  \quad &  \\
  & -4y+z \ge 1 \qquad   \qquad \qquad  \quad &  \\
  & -y+z \ge 1.25 \qquad   \qquad \qquad  \quad &  \\
  & -2y+z \ge 1 \qquad   \qquad \qquad  \quad &  \\
  & -3y+z \ge 0.75 \qquad   \qquad \qquad  \quad &  \\
 & y \ge 0  \qquad   \qquad \qquad  \quad &  \\
\end{eqnarray*}
It is straightforward to verify that:
\begin{displaymath}
X^* \equiv (X_2^4, X_3^4, X_4^4, X_1^3, X_2^3, X_3^3 )
= (\frac{3}{5} , \frac{2}{5} , 0 , 0 , 0 , 0 )
\end{displaymath}
and
\begin{displaymath}
(y^*,z^*) = (\frac{1}{20} , \frac{7}{5}  )
\end{displaymath}
are primal and dual feasible solutions, respectively, achieving the same objective function
value of $\frac{7}{5}$. Thus, they are optimal solution, which implies:
\begin{theorem} \label{apxRatioTypeB}
For Type B-instances,
$A2$ is a $(\rho + \epsilon)$-approximation algorithm for
$(2,4)$-UUSC, where $\rho \le \frac{7}{5} = 1.4$.
\end{theorem}
Combining Theorems \ref{apxRatioTypeA} and \ref{apxRatioTypeB}, and using Assumption \ref{opt4sets},
 altogether we obtain:
\begin{theorem}
$A2$ is a $(\rho + \epsilon)$-approximation algorithm for
$(2,4)$-UUSC, where $\rho \le \frac{35}{24} = 1.458333...$ .
\end{theorem}

In the following, we provide an example for which $\rho = \frac{25}{18} = 1.3888...$ .
The instance is of Type A.
Let $|OPT| = 36m$ for any fixed $m$, that is, $OPT$ consists of $36m$ $4-$columns, denoted $O_1,...,O_{36m}$,
covering $n = 144m$ base elements. The construction of a local optimum $APX$ is as follows.
The $4-$rows in $APX$ consist of two sets: In the first one, for each $i=1,...,12m-1$, there is a $4-$row
which intersects (i.e., covers a single element of) the four columns $O_{3i-2},...,O_{3i+1}$, and
there is one additional $4-$row intersecting $O_1,O_{36m-2},O_{36m-1},O_{36m}$. Thus, the first set contains
$12m$ rows. In the second set, for each $i=1,...,3m$, there are \emph{two} $4-$rows: one intersecting
$O_{3i-1},O_{9m+3i-1},O_{18m+3i-1},O_{27m+3i-1}$, and another one intersecting $O_{3i},O_{9m+3i},O_{18m+3i},O_{27m+3i}$.
Thus, the second set contains $6m$ rows, so the total number of $4-$rows in $APX$ is $X_4 = 18m$.

As for the $3-$rows in $APX$, for each $i=1,...,4m$, there is one $3-$row intersecting $O_{3i-1},O_{12m+3i-1},O_{24m+3i-1}$, and another one intersecting $O_{3i},O_{12m+3i},O_{24m+3i}$. Thus, the
total number of $3-$rows is $X_3 = 8m$.

For a given $\epsilon > 0$, taking $m$ large enough ensures that $APX$ is a local optimum.
%The straightforward way to calculate $\rho$ is to note that the number of $2-$sets which $APX$ must use
%to cover the remaining elements is:
%\begin{displaymath}
%X_2 = \frac{1}{2} ( n - 4X_4 - 3X_3 ) = \frac{1}{2} ( 144m - 4 \cdot 18m - 3 \cdot 8m) = 24m \ .
%\end{displaymath}
%This gives:
%\begin{displaymath}
%\rho = \frac{|APX|}{|OPT|} = \frac{X_2 + X_3 + X_4}{|OPT|} = \frac{50m}{36m}=\frac{25}{18}=1.3888... \ .
%\end{displaymath}
%Alternatively,
Using the pricing scheme, it is easily verified that the $12m$ columns: $O_{3i-2}$, $i=1,...,12m$,
are in $O_{2,0}$, and the remaining $24m$ columns are in $O_{2,1}$. Hence $X_{2,0} = \frac{1}{3}$, $X_{2,1} = \frac{2}{3}$.
The corresponding costs are,  by (\ref{cij}),  $c_{2,0} = \frac{3}{2}$,  $c_{2,1} = \frac{4}{3}$. The obtained
approximation ratio is therefore:
\begin{displaymath}
\rho = c_{2,0}X_{2,0} + c_{2,1}X_{2,1}
= \frac{3}{2} \cdot \frac{1}{3} + \frac{4}{3} \cdot \frac{2}{3}
= \frac{25}{18}.
\end{displaymath}

\section{Concluding remarks}  \label{concluding}
In this paper we focused on a special case of the unweighted $k$-set cover problem.  We proposed a new paradigm to approach instances of this problem, and we showed that it gives better results than the previous known algorithms for unweighted $k$-set cover.  Our proof is for a restricted case in which the instance contains all the pairs of elements.  The technical reason to consider this special case is that all previous known improvements over the greedy algorithm have a special treatment of singletons, which makes the algorithms and their analysis much more complicated.  By neglecting this technical problem, we can concentrate on the way to handle the selection of large sets.

In this paper we showed that the non-oblivious local search methodology can outperform the other methods to approximate  unweighted $k$-set cover, and we conjecture that this is the case for the generalized case and not only for $(2,4)$-uniform instances.  We leave as major open problems the tuning of the parameters for the non-oblivious local search algorithm (i.e., the weights used in the objective function of the local search), as well as the analysis of the resulting algorithm for  unweighted $k$-set cover.

\appendix
\section{Proof of Theorem \ref{LP_solution}}
We  prove that the solution for $(P)$ stated in the theorem is optimal,
and then compute its objective function value.
In order to show optimality, we  construct the dual program of $(P)$,
denoted $(D)$, provide a feasible solution to it, and then use a complementary
slackness argument. By the complementary slackness, we  conclude that both
solutions are optimal. Then, we  compute the objective function
value of the primal solution. We start by constructing $(D)$. The dual decision variables are:
\begin{itemize}
\item $\beta_{p+1},...,\beta_k$ - correspond to the set of constraints (\ref{primalCons1}),
\item $\beta_{p}$ - corresponds to constraint (\ref{primalCons2}),
\item $\gamma_{p+1},...,\gamma_k$ - correspond to the set of constraints (\ref{primalCons3}).
\end{itemize}
The dual program is:\\

\noindent {\bf Program} {\boldmath $(D)$}
\begin{eqnarray}
\min & \sum_{i=p}^k \beta_i \qquad \qquad \qquad \qquad \qquad   &  \nonumber\\
\label{dualCons1}
\textrm{s.t.} &
\beta_k - j \gamma_k \ge \frac{j}{k} \qquad \qquad \qquad \quad &  j=1,...,k-1  \\
\label{dualCons2}
& \beta_k - (k-2) \gamma_k \ge 1 \qquad \qquad  \quad &  \\
\label{dualCons3} & \beta_i - j \gamma_i + j \gamma_{i+1} \ge
\frac{j}{i(i+1)} \qquad
& i=p+1,...,k-1, \quad j=1,...,i-1  \\
\label{dualCons4} & \beta_i - (i-2) \gamma_i + i \gamma_{i+1} \ge
\frac{1}{i+1}  \quad
&  i=p+1,...,k-1   \\
\label{dualCons5}
&  \beta_p + p \gamma_{p+1} \ge \frac{1}{p+1} \quad \qquad \qquad  &    \\
 &  \beta_i, \gamma_j \ge 0 \qquad \qquad \qquad \qquad \quad  &  i=p,...,k, \quad j=p+1,...,k \ . \nonumber \\ \nonumber
\end{eqnarray}
For this LP, the primal variables $a_{k,1},...,a_{k,k-1}$ correspond to the set of constraints (\ref{dualCons1}).
$a_{k,k}$ corresponds to (\ref{dualCons2}). $a_{i,j}$, $i=p+1,...,k-1$, $j=1,...,i-1$ correspond to
(\ref{dualCons3}). $a_{i,i}$, $i=p+1,...,k-1$ correspond to (\ref{dualCons4}),
and $a_{p,p}$ corresponds to (\ref{dualCons5}).

The dual solution is the following (it is the same for the two cases distinguished in $(P)$, depending
on the parity of $k-p$):
\begin{equation} \label{dual_solution}
\begin{array}{llll}
\label{dualSol}
\gamma_k  &=  & \frac{1}{k(k-1)}  &  \\
\gamma_{k-1} &= &  0  &  \\
\gamma_i &= & \gamma_{i+2} + \frac{2}{i(i+1)(i+2)} \ , &  \textrm{for all } i=p+1,...,k-2    \\
\beta_k &= & 1 + (k-2)\gamma_k  &   \\
\beta_i &= & \frac{1}{i+1} - i\gamma_{i+1} + (i-2)\gamma_i \ , & \textrm{for all } i=p+2,...,k-1   \\
\beta_{p+1} &= & \frac{p}{(p+1)(p+2)} + p\gamma_{p+1} -p\gamma_{p+2}  &   \\
\beta_p &= & \frac{1}{p+1}-p\gamma_{p+1} \ . &
\end{array}
\end{equation}

We proceed to verify that the primal and dual solutions which we constructed
are indeed feasible. In addition, we identify the set of tight constraints.
\begin{lemma}  \label{primalFeasTight}
(Primal Feasibility)\\
The primal solution stated in Theorem \ref{LP_solution}
is feasible for $(P)$. Moreover, the set of constraints
(\ref{primalCons1}),(\ref{primalCons2}) and (\ref{primalCons3}) are tight,
except for
(\ref{primalCons3}) for the value of $i=k-1$ when $k-p$ is odd.
\end{lemma}
\begin{proof}
%Consider the primal solution stated in Theorem \ref{LP_solution}.
First, it is trivial that the non-negativity constraints are satisfied since our
solution is $0/1$. The set of constraints (\ref{primalCons1}) are satisfied, and
tight, since for each $i=p+1,...,k$, there exists exactly one $j$ index such that
$a_{i,j}=1$ and for all other $j$ values $a_{i,j}=0$. Similarly, the constraint
(\ref{primalCons2}) is tight, as $a_{p,p}=1$.
As for the set of constraints (\ref{primalCons3}),  which for convenience we rewrite as:
%$\sum_{j=1}^{i-1} j a_{i-1,j} - \sum_{j=1}^{i-1} j a_{i,j} - (i-2)a_{i,i}
% \le   0$, $i=p+1,...,k $, we
\begin{displaymath}
\sum_{j=1}^{p+l} j a_{p+l,j} - \sum_{j=1}^{p+l} j a_{p+l+1,j} - (p+l-1)a_{p+l+1,p+l+1}  \le   0 \ ,
\qquad  l=0,...,k-p-1 \ ,
\end{displaymath}
we distinguish:
\begin{itemize}
\item{\bf Case 1:}
$k-p$ even:
\begin{itemize}
\item{}
For even values of $l$, $0 \le l \le k-p-2$,           % $l=0,2,...,k-p-2$,
the first sum is $p+l$ since $a_{p+l,p+l}=1$ (and all other terms are zero),
the second sum is also $p+l$ because $a_{p+l+1,p+l}=1$,
and the last term is zero. Thus, the left-hand side is zero and the constraint
is tight.
\item{}
For odd values of $l$, $1 \le l \le k-p-1$, %$l=1,3,...,k-p-1$,
the first sum is $p+l-1$ since $a_{p+l,p+l-1}=1$, the second sum is
zero, and the last term is $p+l-1$ since $a_{p+l+1,p+l+1}=1$. The
constraint is tight.
\end{itemize}
\item{\bf Case 2:}
$k-p$ odd:
\begin{itemize}
\item{}
For even values of $l$, $0 \le l \le k-p-3$,            % $l=0,2,...,k-p-3$,
the first sum is $p+l$ since $a_{p+l,p+l}=1$,
the second sum is $p+l$ because $a_{p+l+1,p+l}=1$,
and the last term is zero. The constraint is tight.
\item{}
For odd values of $l$, $1 \le l \le k-p-4$           % $l=1,3,...,k-p-4$,
the first sum is $p+l-1$ since $a_{p+l,p+l-1}=1$,
the second sum is zero, and the last term is $p+l-1$ since $a_{p+l+1,p+l+1}=1$.
The constraint is tight.
\item{}
%It remains to verify feasibility for $l=k-p-2,$ and $k-p-1$.
For $l=k-p-1$,
the first sum is $k-2$ since $a_{k-1,k-2}=1$, the second sum is zero,
and the last term is $k-2$ since $a_{k,k}=1$. The constraint is tight.
\item{}
For $l=k-p-2$,
the first sum is $k-3$ since $a_{k-2,k-3}=1$, the second sum is $k-2$
because $a_{k-1,k-2}=1$ and the last term is zero. Thus, the left-hand
side is $-1$, so the constraint is satisfied but not tight. Observe that this case
corresponds to the value $i=k-1$ in the original formulation (\ref{primalCons3}).
\end{itemize}
\end{itemize}
\end{proof}

We next consider the feasibility of the
dual solution given by (\ref{dual_solution}).
First of all, it is trivial that $\gamma_i \ge 0$, $i=p+1,...,k$.
Next, by straightforward substitution,  it is easily verified that the dual constraints
(\ref{dualCons2}) and (\ref{dualCons5}) are tight. For the other
constraints, we  use the following auxiliary calculations:
\begin{lemma}
For $i=p+1,...,k-1:$
\begin{equation} \label{sum_2_consec_gams}
\gamma_i + \gamma_{i+1} = \frac{1}{i(i+1)} \ , %\qquad \textrm{for all } i=p+1,...,k-1 %\\
\end{equation}
\begin{equation} \label{diff_2_consec_gams}
\gamma_{i+1} - \gamma_i \le  \frac{1}{i(i+1)} \ . %\quad  % \textrm{for all } i=p+1,...,k-1 %\\
\end{equation}
\end{lemma}
\begin{proof}
The first part is proved by induction: The case $i=k-1$ is immediate
since $\gamma_{k-1}=0$ and $\gamma_{k}=\frac{1}{k(k-1)}$. Assume that
(\ref{sum_2_consec_gams}) holds for $i$, $p+2 \le i \le k-1$. Then
for $i-1$, we have:
\begin{displaymath}
\begin{array}{lll}
\gamma_{i-1}+\gamma_{i}-\frac{1}{(i-1)i} &  = & \gamma_{i+1}+\frac{2}{(i-1)i(i+1)}+\gamma_{i}-\frac{1}{(i-1)i} \\
 &  = & \gamma_{i}+\gamma_{i+1} - \frac{1}{i(i+1)} \\
 & = & 0 \ ,
\end{array}
\end{displaymath}
where the last equality holds by the induction hypothesis.
For the second part, observe that
\begin{displaymath}
\frac{1}{i(i+1)} - \gamma_{i+1} + \gamma_i  \ge \frac{1}{i(i+1)} - \gamma_{i+1} - \gamma_i = 0 \ ,
\end{displaymath}
where the inequality holds since $\gamma_i \ge 0$ and the equality is by (\ref{sum_2_consec_gams}).
The result follows.
\end{proof}
\begin{lemma}  \label{dualFeasTight}
(Dual Feasibility)\\
The dual solution defined by (\ref{dualSol}) is feasible for $(D)$.
Moreover, the tight constraints are
(\ref{dualCons2}),(\ref{dualCons4}),(\ref{dualCons5}), and
(\ref{dualCons3}) for $i$ and $j$ values such that $j=i-1$.
\end{lemma}
\begin{proof}
We first identify the tight constraints in $(D)$. Consider the
set of constraints (\ref{dualCons4}). For $i$ values
$i=p+2,...,k-1$, it is easily seen that they are tight, by the definition
of $\beta_i$. For $i=p+1$, substituting the definition of
$\beta_{p+1}$, we obtain
\begin{displaymath}
\frac{p}{(p+1)(p+2)} + \gamma_{p+1} + \gamma_{p+2}  \ge \frac{1}{p+2} \ .
\end{displaymath}
By (\ref{sum_2_consec_gams}), the inequality is tight.

Consider the set of constraints (\ref{dualCons3}), for $j$ values $j=i-1$.
From $\beta_{p+1}$'s definition, it follows immediately that the case $i=p+1$
(hence $j=p$) is tight.
For $i \ge p+2$, substituting $\beta_i$'s definition yields:
\begin{displaymath}
\frac{1}{i+1} - \gamma_i - \gamma_{i+1} \ge \frac{i-1}{i(i+1)} \ .
\end{displaymath}
Again, (\ref{sum_2_consec_gams}) yields that it is tight.
   \\

We are done identifying the tight dual constraints. We now turn to
verify feasibility for the rest of the constraints. Consider the set
of constraints (\ref{dualCons1}). Substituting the definitions of
$\beta_k$ and $\gamma_k$, we obtain:
\begin{displaymath}
1 + \frac{k-j-2}{k(k-1)} \ge \frac{j}{k} \ , \qquad j=1,...,k-1 \ .
\end{displaymath}
The inequality clearly holds for $j \le k-2$. For $j=k-1$ it
evaluates to $1 - \frac{1}{k(k-1)} \ge \frac{k-1}{k}$, which holds, as
$k \ge 2$.
  \\
\noindent
Consider the set of constraints (\ref{dualCons3}), for $i=p+1,...,k-1$ and
$j=1,...,i-2$ (we have shown that cases for the values $j=i-1$ are tight).
For $i=p+2,...,k-1$, we evaluate $\beta_i$'s definition to obtain:
\begin{displaymath}
\begin{array}{lll}
\beta_i & = &  \frac{1}{i+1} - i \gamma_{i+1} + (i-2) \gamma_{i}   \\
& = &  \frac{1}{i+1} - \frac{1}{i(i+1)} - (i-1) \gamma_{i+1}    + (i-1) \gamma_{i}  \\
& = &  \frac{i-1}{i(i+1)} - (i-1) \gamma_{i+1}    + (i-1) \gamma_{i}  \ge   0 \ , \\
\end{array}
\end{displaymath}
where the third equality follows from (\ref{sum_2_consec_gams}) and
the inequality follows from (\ref{diff_2_consec_gams}). This proves
that for $i=p+2,...,k-1$, the constraints (\ref{dualCons3}) hold,
and also that $\beta_i \ge 0$. For $i=p+1$ (and $j=1,...,p-1$), we
evaluate $\beta_{p+1}$'s definition:
\begin{displaymath}
\beta_{p+1}  =   \frac{p}{(p+1)(p+2)} + p \gamma_{p+1}- p \gamma_{p+2}  \ge 0  \ ,
\end{displaymath}
where again, the inequality follows from
(\ref{diff_2_consec_gams}). This establishes that (\ref{dualCons3})
holds for $p+1$ $(j=1,...,p-1)$ and that $\beta_{p+1} \ge 0$.

It remains to show that $\beta_p$ and $\beta_k$ are nonnegative. As
$\gamma_k \ge 0$ and $k \ge 2$, it immediately follows from the
definition that $\beta_k \ge 0$. For $\beta_p$, we use
(\ref{sum_2_consec_gams}) to obtain:
\begin{displaymath}
\begin{array}{lll}
\beta_p & = &  \frac{1}{p+1} - p \gamma_{p+1}
         \ge   \frac{1}{p+1} - p ( \gamma_{p+1} + \gamma_{p+2} )  \\
& = &  \frac{1}{p+1} - \frac{p}{(p+1)(p+2)} =  \frac{2}{(p+1)(p+2)}  \ge 0 \ . \\
\end{array}
\end{displaymath}
\end{proof}  \\

We now show that the solutions $(a_{p,p},$
$a_{p+1,1},...,a_{p+1,p+1},$ $a_{p+2,1},...,a_{p+2,p+2},$
$a_{k,1},...,a_{k,k} )$ and $(\beta_p,...,\beta_k,$
$\gamma_{p+1},...,\gamma_k)$
satisfy the complementary slackness
conditions with respect to programs $(P)$ and $(D)$. Thus, they are both optimal.
\begin{lemma} (Primal and Dual Optimality) \\
The primal solution $(a_{p,p},$ $a_{p+1,1},...,a_{p+1,p+1},$
$a_{p+2,1},...,a_{p+2,p+2},$ $a_{k,1},...,a_{k,k} )$
stated in Theorem \ref{LP_solution} is optimal for $(P)$.
The dual solution $(\beta_p,...,\beta_k,$ $\gamma_{p+1},...,\gamma_k)$ defined by
(\ref{dualSol}) is optimal for $(D)$.
\end{lemma}
\begin{proof}
Consider $(P)$. By Lemma \ref{primalFeasTight}, all
of $(P)$'s constraints are tight except for (\ref{primalCons3}) for
the value of $i=k-1$ (and when $k-p$ is odd). But the dual variable
corresponding to this constraint, $\gamma_{k-1}$, is zero.
Hence, all primal complementary slackness conditions are satisfied.

Consider $(D)$. From Lemma \ref{dualFeasTight}, it follows that the
constraints which are not tight are (\ref{dualCons1}), and
(\ref{dualCons3}) for the case $j=1,...,i-2$. The primal variables
corresponding to (\ref{dualCons1}) are $a_{k,1},...,a_{k,k-1}$, and
are all zeros. Hence the conditions are satisfied. The variables
corresponding to (\ref{dualCons3}) for $j=1,...,i-2$ are $a_{i,j}$,
$i=p+1,...,k-1$, $j=1,...i-2$. All of them are zeros, so again, the
conditions are satisfied. Therefore, all dual complementary
slackness conditions are satisfied.
Since the complementary slackness conditions hold, the primal (as well as the dual) solution
is optimal.
\end{proof}  \\

We now compute the primal objective function, denoted $POF$. This time we distinguish four cases,
depending on the parity of both $k$ and $p$. In each case, we substitute the primal solution in the
objective function.
\begin{itemize}
\item{\bf Case 1:} $p$ even, $k$ even (the term $1$ is for $a_{k,k}$):\\
\begin{displaymath}
\begin{array}{lll}
POF & = & 1 + \sum_{j=\frac{p}{2}}^{ \frac{k}{2} - 1 } \left(  \frac{2j}{2j(2j+1)}+ \frac{2j}{ (2j+1)(2j+2)} \right) \\
&  = & 1 + \sum_{j=\frac{p}{2}}^{ \frac{k}{2} - 1 } \left(  \frac{1}{2j+1} + \frac{j}{ (2j+1)(j+1)} \right) \\
& = & 1+ \sum_{j=\frac{p}{2}}^{ \frac{k}{2} - 1 } \frac{1}{j+1}  \\
& = & H_\frac{k}{2} - H_{\frac{p}{2}} + 1 \ . \\
\end{array}
\end{displaymath}

\item{\bf Case 2:} $p$ even, $k$ odd (the first two terms are for $a_{k,k}$, $a_{k-1,k-2}$ respectively):\\
\begin{displaymath}
\begin{array}{lll}
POF & = & 1 + \frac{k-2}{k(k-1)}
+ \sum_{j=\frac{p}{2}}^{ \frac{k-3}{2}  } \left(  \frac{2j}{2j(2j+1)} + \frac{2j}{ (2j+1)(2j+2)} \right) \\
& = & 1 + \frac{1}{k} - \frac{1}{k(k-1)}
+ \sum_{j=\frac{p}{2}}^{ \frac{k-3}{2} } \left(  \frac{1}{2j+1} + \frac{j}{ (2j+1)(j+1)} \right) \\
& = & 1 + \frac{1}{k} - \frac{1}{k(k-1)}
+ \sum_{j=\frac{p}{2}}^{ \frac{k-3}{2}  }   \frac{1}{j+1}  \\
&  = &  H_\frac{k-1}{2} - H_\frac{p}{2}   + 1 + \frac{1}{k} - \frac{1}{k(k-1)} \ . \\
\end{array}
\end{displaymath}

\item{\bf Case 3:} $p$ odd, $k$ even (the first two terms are for $a_{k,k}$, $a_{k-1,k-2}$ respectively):\\
\begin{displaymath}
\begin{array}{lll}
POF & = & 1  + \frac{k-2}{k(k-1)}
+ \sum_{j=\frac{p-1}{2}}^{ \frac{k}{2} - 2 } \left(  \frac{2j+1}{(2j+1)(2j+2)}+ \frac{2j+1}{ (2j+2)(2j+3)} \right) \\
& = & 1 + \frac{1}{k} - \frac{1}{k(k-1)} + \sum_{j=\frac{p-1}{2}}^{ \frac{k}{2} - 2 }  \frac{2}{2j+3}
\\
& = & 1 + \frac{1}{k} - \frac{1}{k(k-1)} + 2 \sum_{j=\frac{p+1}{2}}^{ \frac{k}{2} - 1 }  \frac{1}{2j+1}
 \\
& = & 2(H_{k} - H_{p+1}) - H_\frac{k}{2}  + H_\frac{p+1}{2} + 1 +  \frac{1}{k} - \frac{1}{k(k-1)}
 \ . \\
\end{array}
\end{displaymath}
where the last equality follows from the straightforward identity:
\begin{equation}  \label{oddHarmonic}
\sum_{j=l}^r  \frac{1}{2j+1} = H_{2r+2} - H_{2l} - \frac{1}{2}(H_{r+1}-H_l) \ , \ \quad \textrm{for all }  l \le r \ .
\end{equation}
\item{\bf Case 4:} $p$ odd, $k$ odd (the term $1$ is for $a_{k,k}$):\\
\begin{displaymath}
\begin{array}{lll}
POF & = &  1  + \sum_{j=\frac{p-1}{2}}^{ \frac{k-3}{2} } \left(  \frac{2j+1}{(2j+1)(2j+2)}+ \frac{2j+1}{ (2j+2)(2j+3)} \right)
\\
& = & 1  + \sum_{j=\frac{p-1}{2}}^{ \frac{k-3}{2} } \frac{2}{2j+3} \\
& = & 1 + 2\sum_{j=\frac{p+1}{2}}^{ \frac{k-1}{2} } \frac{1}{2j+1} \\
& = & 2(H_{k+1} - H_{p+1}) - H_\frac{k+1}{2}  + H_\frac{p+1}{2} + 1 \ . \\
\end{array}
\end{displaymath}
where again, we used (\ref{oddHarmonic}).
\end{itemize}
This completes the proof of Theorem \ref{LP_solution}.
\endpf

\end{document}